\documentclass[12pt]{article}
\usepackage[cp1250]{inputenc}
\usepackage[verbose,a4paper,tmargin=0.5in,lmargin=1in,rmargin=1in]{geometry}
\usepackage[pdftex]{graphicx}
\usepackage{amsfonts}
\usepackage{indentfirst}
\usepackage{latexsym}
\usepackage{amsmath}
\usepackage{amsthm}
\usepackage{amssymb}
\usepackage{psfrag}
\usepackage[mathscr]{eucal} 
\usepackage{color}
\usepackage{cite}
\usepackage{longtable}
\usepackage{verbatim}
\usepackage{enumitem}
\usepackage[font=normalsize]{caption}
\usepackage{bigints}

\usepackage{graphicx}
\usepackage[all]{xy}
\usepackage{float}
\usepackage{afterpage}

\usepackage{physics}

\newtheorem{Twierdzenie}{Theorem}[section]
\newtheorem{Definicja}{Definition}[section]
\newtheorem{Lemat}{Lemma}[section]
\newtheorem{Uwaga}{Remark}[section]
\newtheorem{Wniosek}{Corollary}[section]

\newfloat{Scheme}{tbp}{muj}

\title{On Walker and para-Hermite Einstein spaces.}

\author{$\textrm{Adam Chudecki}^{a,b,c}$}

\begin{document}
\sloppy

\maketitle

$^a$ Center of Mathematics and Physics, Lodz University of Technology, Al. Politechniki 11, 90-924 Lodz, Poland. 
\newline
$^b$ Institute of Physics, Faculty of Technical Physics, Information Technology and Applied Mathematics, Lodz University of Technology, Wólcza\'nska 217/221, 93-005 Lodz, Poland.
\newline
$^c$ adam.chudecki@p.lodz.pl
\newline
\newline
\textbf{Abstract}. A special class of (complex) para-Hermite Einstein spaces is analyzed. For this class of spaces the self-dual Weyl tensor is type-[D] in the Petrov-Penrose classification. The anti-self-dual Weyl tensor is algebraically degenerate, equivalently, there exists an anti-self-dual congruence of null strings. It is assumed that this congruence is parallely propagated. Thus, the spaces are not only para-Hermite but also Walker. A classification of the spaces according to three criteria is given. Finally, explicit metrics of all admitted Petrov-Penrose types are found.


\section{Introduction}
\setcounter{equation}{0}

\subsection{Background}

Real 4-dimensional neutral spaces (i.e., smooth 4D-manifolds equipped with a smooth metric of neutral signature $(++--)$) play an important role in mathematics and theoretical physics. For example, \textsl{Walker spaces}, \textsl{para-Hermite spaces} (\textsl{pH-spaces}), \textsl{para-K\"ahler spaces} (\textsl{pK-spaces)} and their Einstein counterparts (\textsl{pHE-spaces}, \textsl{pKE-spaces}) belong to this class. Walker spaces are equipped with a totally null 2D-distribution which is integrable and parallely propagated \cite{Walker}. PH-spaces and pK-spaces admit a pair of integrable and totally null 2D-distributions. Such spaces were defined in the fifties \cite{Libermann} and since then they have been intensively studied \cite{Flaherty}. Recently it has been discovered that pKE-spaces are essential in geometrically and physically important issues, like non-integrable twistor distributions, (2,3,5)-distributions with maximal algebra of infinitesimal symmetries, rolling without slipping or twisting, self-dual (SD) and anti-self-dual (ASD) spaces \cite{Alexandrov, An, Bor, Bor_Makhmali_Nurowski}. 

Motivated by these results we dealt with pKE-spaces and pHE-spaces in \cite{Chudecki_Examples} where we found many examples. The examples, however, were not general and the question whether vacuum Einstein equations with cosmological constant in pKE-spaces and pHE-spaces can be solved in all generality, remained open. Later we understood that a key to success lies in equations of congruences of null strings. [Congruences of null strings are complex counterparts of congruences of null geodesics. They were analyzed in the seventies \cite{Plebanski_Hacyan_Sachs,Plebanski_Rozga} and they led to the discovery of hyperheavenly spaces ($\mathcal{HH}$-spaces) \cite{Plebanski_Robinson_2}]. We realized that vacuum Einstein equations with cosmological constant should not be attacked before the equations of congruences of null strings are solved. Being aware of this fact, we returned to pKE-spaces. Eventually, we found general metrics of all pKE-spaces with algebraically degenerate ASD Weyl tensor. Results published in \cite{Chudecki_Ref_1,Chudecki_Ref_2} are significant contribution to the field of exact solutions of pKE-spaces. In fact, only an algebraically general case has not been solved yet. However, the question of general solutions of pHE-spaces remained open.

In this paper we try to, at least partially, fill this gap. We deal with a special class of pHE-spaces. For all pHE-spaces the SD Weyl tensor is type-[D] in the Petrov-Penrose classification (see Section \ref{section_preliminaries}). The ASD Weyl tensor could be arbitrary but in what follows we assume that it is algebraically degenerate. Equivalently, there exists an integrable, ASD, totally null, 2D-distribution. We assume  that this distribution is parallely propagated. Thus, all spaces considered in the article are also Walker spaces. In other words, we deal with spaces of the types $[\textrm{D}]^{ee} \otimes [\textrm{deg}]^{n}$ (the meaning of this symbol is explained in Section \ref{section_preliminaries}). Our goal is to find general metrics of such spaces. To achieve this goal we use the theory of hyperheavenly spaces.

Theory of $\mathcal{HH}$-spaces is a powerful tool in the general theory of relativity and in various branches of differential geometry. $\mathcal{HH}$-spaces are defined as 4-dimensional, complex, Einstein spaces with cosmological constant $\Lambda$ equipped with a holomorphic metric such that the SD (or ASD) Weyl tensor is algebraically degenerate\footnote{The only case which is not covered by the theory of $\mathcal{HH}$-spaces is the type $[\textrm{I}] \otimes [\textrm{I}]$.}. Vacuum Einstein field equations with cosmological constant in $\mathcal{HH}$-spaces were reduced to a single, partial, second order, strongly nonlinear differential equation called \textsl{the hyperheavenly equation ($\mathcal{HH}$-equation)} for a holomorphic function $W$ (called \textsl{the key function}) which completely determines the metric (Pleba\'nski, Robinson, \cite{Plebanski_Robinson_2}). An integration of vacuum Einstein equations with cosmological constant to a single differential equation is one of the most outstanding results in the field of complex methods in the general theory of relativity. 

In what follows we specify $\mathcal{HH}$-spaces to be of the types  $[\textrm{D}]^{ee} \otimes [\textrm{deg}]^{n}$ and we solve the corresponding $\mathcal{HH}$-equation in all the cases. PHE-metrics are obtained as real neutral slices of the corresponding complex metrics. Fortunately, real neutral slices of complex metrics can be easily found.

All considerations are purely local. We use the spinorial formalism in Infeld - Van der Waerden - Pleba\'nski notation. All coordinates are complex, all functions are holomorphic.

\subsection{Structure of the article}

Our paper is organized, as follows. Section \ref{section_preliminaries} is devoted to geometrical preliminaries. It consists of a brief introduction to geometrical structures (like congruences of null strings and congruences of null geodesics), the Petrov-Penrose classification and para-Hermite spaces. We also introduce several abbreviations used in the rest of the paper. In order not to make the article too long, we skipped an introduction to the spinorial formalism. It is cordially suggested to the Reader of this paper to read Sections 1.3 and 2 of \cite{Chudecki_Ref_1} where all the abbreviations were introduced and illustrated with examples and a brief introduction to the spinorial formalism in Infeld - Van der Waerden - Pleba\'nski approach was presented. Readers who are familiar with our previous cycle of publications \cite{Chudecki_Ref_1, Chudecki_Ref_2} could skip Section 2 and start studying the article from Section \ref{section_HH_spaces}. 

In Section \ref{section_HH_spaces} we explain how pHE-spaces fit into the theory of expanding $\mathcal{HH}$-spaces. Formulas for the metric, null tetrad and Weyl curvature coefficients of expanding $\mathcal{HH}$-spaces for which SD  Weyl tensor is type-[D] are given. Some remarks about congruences of SD and ASD null strings and their intersections are posted. We solve equations of a congruence of ASD null strings and we find a form of the key function (\ref{Key_function_step_3}). We prove that the $\mathcal{HH}$-equation (\ref{HH_equation_ogolne}) reduces to the system (\ref{HH_resztki_1})-(\ref{HH_resztki_2}) which is "the core" of the further analysis. Finally, we discuss criteria for the ASD Weyl tensor to be of a specific Petrov-Penrose type.

Sections \ref{section_Dee_x_IIn} - \ref{section_Dee_x_Nn} are devoted to exact solutions of different Petrov-Penrose types. We always follow the same pattern. First we deal with spaces for which congruences of null geodesics are all expanding and twisting. Then we focus on a simpler class for which one nonexpanding and nontwisting congruence of null geodesics is admitted. Finally, we analyze even more special solutions with 2D and 3D-algebra of infinitesimal symmetries (it is explained in Section \ref{section_HH_spaces} that any pHE-space is equipped with at least one Killing vector). 

Some concluding remarks with a sketch of possible directions for further investigations close the paper.

\subsection{Summary of main results}

To convince the Reader that the results of the paper constitute a significant progress in the field of exact solutions of pHE-spaces, we make a brief comparison with our first paper devoted to pHE-spaces \cite{Chudecki_Examples}. Results published in \cite{Chudecki_Examples} are far away from being general. It is so because of an additional assumption which we made (see Eq. (3.12) of \cite{Chudecki_Examples}). In what follows our approach is general and it does not involve any "ad hoc" assumptions. Moreover, the current paper contains a complete geometrical classification of the results, which was not done in \cite{Chudecki_Examples}.

The main results of the paper are gathered in Sections \ref{section_Dee_x_IIn} - \ref{section_Dee_x_Nn}. Consider the metric (\ref{metryka_TypII_pp_pp_ostateczna}). It is the general metric of a space of the type $\{ [\textrm{D}]^{ee} \otimes [\textrm{II}]^{n}, [++,++] \}$. It depends on one function of two variables which satisfies an Abel differential equation of the first kind. Our previous example of such a metric (\cite{Chudecki_Examples}, the metric (3.31) with $f=0$) depends on one function of one variable only. This is a significant progress in generality of considerations.

Similarly, the metric (\ref{metryka_TypIII_pp_pp_ostateczna}) is the general metric of a space of the type $\{ [\textrm{D}]^{ee} \otimes [\textrm{III}]^{n}, [++,++] \}$ and it depends on two functions of one variable. It is more general then its counterpart of \cite{Chudecki_Examples} (the metric (3.61), one function of one variable).

A special attention should be paid on the metric (\ref{metryka_TypN_pp_pp_ostateczna}). It is the general metric of a space of the type $\{ [\textrm{D}]^{ee} \otimes [\textrm{N}]^{n}, [++,++] \}$. In this case we cannot compare it with any other metric because according to our best knowledge, (\ref{metryka_TypN_pp_pp_ostateczna}) is the first example of a pHE-space for which the ASD Weyl tensor is type-[N]. In our opinion this fact makes Section \ref{section_Dee_x_Nn} the most valuable part of the paper. 

Results concerning types $ [\textrm{D}]^{ee} \otimes [\textrm{D}]^{nn}$ (Section \ref{section_Dee_x_Dnn}) are less original. The metrics (\ref{metryka_TypD_ostateczna_pppppppp}) and (\ref{metryka_TypD_ostateczna_ppmmmmpp}) belong to pKE class. Such solutions were found in \cite{Bor_Makhmali_Nurowski} in all generality. However, our approach is slightly different and it is enriched with an additional geometrical subclassification of such spaces.

\subsection{Remarks about real slices}

In the paper we mainly deal with complex 4D-spaces: metrics are holomorphic and coordinates are complex. Real pHE-spaces can be, however, quite easily obtained from the generic complex results. The metrics should be only considered as real ones (with real coordinates and real smooth functions). Also, the Petrov-Penrose types should be replaced by their real neutral counterparts, as below
\begin{equation}
\nonumber
[\textrm{D}] \rightarrow [\textrm{D}_{r}], \ [\textrm{II}] \rightarrow [\textrm{II}_{r}], \ [\textrm{III}] \rightarrow [\textrm{III}_{r}], \ [\textrm{N}] \rightarrow [\textrm{N}_{r}]
\end{equation}

One can also obtain neutral slices with more "exotic" Petrov-Penrose types ($[\textrm{D}_{c}]$, $[\textrm{II}_{rc}]$), but a corresponding procedure is more complicated and we do not deal with such slices in this paper.  
[Subtleties of the Petrov-Penrose classification of conformal curvature in real neutral spaces have been discussed in many papers, see, e.g., \cite{Bor_Makhmali_Nurowski} or Section 2.2 of \cite{Chudecki_Ref_1}].

The only candidates which could possibly admit real Lorentzian slices are spaces of the types $[\textrm{D}]^{ee} \otimes [\textrm{D}]^{nn}$ (Section \ref{section_Dee_x_Dnn}), because both SD and ASD Weyl spinors are of the same Petrov-Penrose type. However, within the types $[\textrm{D}]^{ee} \otimes [\textrm{D}]^{nn}$ congruences of SD null strings are expanding while congruences of ASD null strings are nonexpanding. Such spaces do not admit Lorentzian slices. Hence, the metrics (\ref{metryka_TypD_ostateczna_pppppppp}) and (\ref{metryka_TypD_ostateczna_ppmmmmpp}) are interesting examples of two-sided type-[D] complex spaces which do not admit Lorentzian slices.


\renewcommand{\arraystretch}{1.5}
\setlength\arraycolsep{2pt}

\section{Preliminaries}
\label{section_preliminaries}
\setcounter{equation}{0}

\subsection{Petrov-Penrose classification}

A spinorial image of the SD Weyl tensor is called \textsl{the SD Weyl spinor}. It is an undotted 4-index spinor  totally symmetric in all indices, $C_{ABCD} = C_{(ABCD)}$. It is well known that the SD Weyl spinor can be decomposed into a symmetric product of 1-index spinors, $C_{ABCD} = a_{(A} b_{B} c_{C} d_{D)}$. Spinors $a_{A}$, $b_{B}$, $c_{C}$ and $d_{D}$ are called \textsl{undotted Penrose spinors}. The spinors $a_{A}$, $b_{B}$, $c_{C}$ and $d_{D}$ are mutually linearly independent, in general. In such a case the SD Weyl spinor is \textsl{algebraically general}. Otherwise, it is \textsl{algebraically special}. \textsl{The Petrov-Penrose classification} consists in distinguishing the possible coincidences between the spinors $a_{A}$, $b_{B}$, $c_{C}$ and $d_{D}$. More precisely, we have
\begin{eqnarray}
\textrm{type [I]} : \ && \ C_{ABCD} = a_{(A} b_{B} c_{C} d_{D)}
\\ \nonumber
\textrm{type [II]}: \ && \ C_{ABCD} = a_{(A} a_{B} b_{C} c_{D)}
\\ \nonumber
\textrm{type [D]}: \ && \  C_{ABCD} = a_{(A} a_{B} b_{C} b_{D)}
\\ \nonumber
\textrm{type [III]}: \ && \  C_{ABCD} = a_{(A} a_{B} a_{C} b_{D)}
\\ \nonumber
\textrm{type [N]}: \ && \  C_{ABCD} = a_{A} a_{B} a_{C} a_{D}
\\ \nonumber
\textrm{type [O]}: \ && \  C_{ABCD} = 0 
\end{eqnarray} 

The Petrov-Penrose classification of the ASD Weyl spinor $C_{\dot{A}\dot{B}\dot{C}\dot{D}}$ is analogous except the fact that the ASD Weyl spinor is decomposable into a symmetric product of 1-index \textsl{dotted Penrose spinors}. In complex spaces both SD and ASD Weyl spinors are independent. Thus, a commonly accepted symbol in which a type of conformal curvature is encoded reads 
\begin{equation}
\label{symbol_typy}
[\textrm{SD}_{\textrm{type}}] \otimes [\textrm{ASD}_{\textrm{type}}]
\end{equation}
where $\textrm{SD}_{\textrm{type}}, \textrm{ASD}_{\textrm{type}} = \{ \textrm{I}, \textrm{II}, \textrm{D}, \textrm{III}, \textrm{N}, \textrm{O}  \}$.

The Petrov-Penrose classification of the SD and ASD Weyl spinors in real neutral spaces is a little more complicated (compare, e.g., the Table 1 of \cite{Chudecki_Ref_1}).

The Petrov-Penrose classification is the first of three criteria which we use for classification of pHE-spaces.

\subsection{Congruences of null strings}

Geometrical structures which are fundamental in our approach to pHE-spaces are \textsl{congruences (foliations) of SD (or ASD) null strings}. Below we present only a brief introduction to these structures, for more details see, e.g., \cite{Plebanski_Rozga, Chudecki_notes_on_congr}.

To understand what a congruence of SD null strings is, we introduce first a 2D-holomorphic distribution $\{ m_{A} a_{\dot{B}}, m_{A} b_{\dot{B}} \}$ where $m_{A}$ is an undotted spinor, while $a_{\dot{A}}$ and $b_{\dot{A}}$ are dotted spinors such that $a_{\dot{A}} b^{\dot{A}} \ne 0$. The distribution  $\{ m_{A} a_{\dot{B}}, m_{A} b_{\dot{B}} \}$ is completely integrable in the Frobenius sense, if and only if the spinor $m_{A}$ satisfies the set of equations
\begin{equation}
\label{SD_null_strings_equations}
m^{A} \nabla_{B\dot{B}} m_{A} = m_{B} M_{\dot{B}}
\end{equation}

Eqs. (\ref{SD_null_strings_equations}) we call \textsl{SD null string equations}. If a spinor $m_{A}$ satisfies (\ref{SD_null_strings_equations}) we say that it \textsl{generates a congruence of SD null strings}. Integral manifolds of the distribution $\{ m_{A} a_{\dot{B}}, m_{A} b_{\dot{A}} \}$ are 2D-holomorphic surfaces called \textsl{SD null strings}. A null string is a totally null surface (in the sense that any vector tangent to it is null). A \textsl{congruence (foliation) of SD null strings} is a family of such surfaces\footnote{In Penrose terminology, a SD (ASD) null string is called \textsl{$\alpha$-surface} (\textsl{$\beta$-surface}).}. 

Let us return to Eqs. (\ref{SD_null_strings_equations}). The right hand side of (\ref{SD_null_strings_equations}) is proportional to a dotted spinor $M_{\dot{B}}$ which is called \textsl{an expansion of a congruence}. Its geometrical meaning is fundamental. If $M_{\dot{B}}=0$ holds true then the distribution $\{ m_{A} a_{\dot{B}}, m_{A} b_{\dot{A}} \}$ is parallely propagated. Such congruences were analyzed in a distinguished work by A.G. Walker \cite{Walker} for the first time. Spaces which are equipped with such congruences are called \textsl{Walker spaces}. Congruences for which $M_{\dot{B}}=0$ we call \textsl{nonexpanding} while congruences for which $M_{\dot{B}} \ne 0$ we call \textsl{expanding}. This terminology could be a little misleading, however, it was originally used in the seventies by the pioneers of the $\mathcal{HH}$-spaces theory (J.F. Pleba\'nski, I. Robinson, J.D. Finley III, S. Hacyan, M. Przanowski and others). It was also used by the Author of this paper in his previous works. Hence, it will be used in the rest of the current work. 

A congruence of ASD null strings is defined analogously like a SD one. It is generated by a dotted spinor $m_{\dot{A}}$ with expansion given by an undotted spinor $M_{A}$. \textsl{ASD null string equations} take the form
\begin{equation}
\label{ASD_null_strings_equations}
m^{\dot{A}} \nabla_{B \dot{B}} m_{\dot{A}} = m_{\dot{B}} M_{B}
\end{equation}

Following \cite{Chudecki_Ref_1} we introduce abbreviations
\begin{eqnarray}
\nonumber
\mathcal{C} &-& \textrm{a congruence of null strings}
\\ \nonumber
\mathcal{C}s &-& \textrm{congruences of null strings}
\\ \nonumber
\mathcal{C}_{m^{A}} &-& \textrm{a SD congruence of null strings generated by a spinor } m^{A}
\\ \nonumber
\mathcal{C}_{m^{\dot{A}}} &-& \textrm{an ASD congruence of null strings generated by a spinor } m^{\dot{A}}
\end{eqnarray}

The famous \textsl{complex Goldberg-Sachs theorem} says \cite{Plebanski_Hacyan_Sachs, Przanowski_classification}
\begin{Twierdzenie}
\label{Twierdzenie_Goldberga_Sachsa}
If $(\mathcal{M}, ds^{2})$ is a complex 4-dimensional Einstein space then the following statements are equivalent
\begin{enumerate}[label=(\roman*)]
\item $(\mathcal{M}, ds^{2})$ admits a congruence of self-dual (anti-self-dual) null strings generated by a spinor $m_{A}$ ($m_{\dot{A}}$)
\item self-dual (anti-self-dual) Weyl spinor is algebraically degenerate and $m_{A}$ ($m_{\dot{A}}$) is a multiple undotted (dotted) Penrose spinor
\end{enumerate}
\end{Twierdzenie}
If a 4D-space admits two SD $\mathcal{C}s$, say $\mathcal{C}_{m^{A}}$ and $\mathcal{C}_{n^{A}}$, such that $m_{A} n^{A} \ne 0$, then such congruences are called \textsl{complementary}. Complementary congruences have necessarily the same duality. Hence, we arrive at
\begin{Wniosek}
\label{wniosek_o_type_d_Ein}
Let $(\mathcal{M}, ds^{2})$ be a complex 4-dimensional Einstein space admitting two complementary self-dual (anti-self-dual) congruences of null strings generated by spinors $m_{A}$ and $n_{A}$ ($m_{\dot{A}}$ and $n_{\dot{A}}$). Then the self-dual (anti-self-dual) Weyl spinor is type-[D], $C_{ABCD} = m_{(A} m_{B} n_{C} n_{D)}$ ($C_{\dot{A}\dot{B}\dot{C}\dot{D}} = m_{(\dot{A}} m_{\dot{B}} n_{\dot{C}} n_{\dot{D})}$).
\end{Wniosek}

Properties of $\mathcal{C}s$ is the second of three criteria which we use for classification of pHE-spaces.

\subsection{Para-Hermite spaces}

According to \cite{Przanowski_Formanski_Chudecki} we have the following definition
\begin{Definicja}
\label{Definicja_para_Hermite}
A complex para-Hermite space is a triple $(\mathcal{M}, ds^{2}, K)$ where $\mathcal{M}$ is a 4-dimensional complex manifold, $ds^{2}$ is a nondegenerate holomorphic metric and $K: T \mathcal{M} \rightarrow T \mathcal{M}$ is an endomorphism such that
\begin{eqnarray}
\nonumber
(i) && K^{2} = id_{T \mathcal{M}}, \ \ \textrm{(K is paracomplex)}
\\ \nonumber
(ii) && \pm 1 \textrm{ eigenvalues of } K \textrm{ have rank 2}
\\ \nonumber
(iii) && ds^{2}(KX, KY) = - ds^{2}(X,Y) \ \textrm{for all } X,Y \in T \mathcal{M}, \ \ \textrm{(K is metric compatible)}
\\ \nonumber
(iv) && \textrm{the Nijenhuis tensor } N: T \mathcal{M} \otimes T \mathcal{M} \rightarrow T \mathcal{M},
\\ \nonumber
&& N(X,Y) := K[X, KY] + K[KX, Y] -[KX, KY] - [X,Y] 
\\ \nonumber
&& \textrm{ vanishes, i.e., } N(X,Y)=0 \textrm{ for all } X,Y \in T \mathcal{M}
\end{eqnarray}
\end{Definicja}
\begin{Definicja}
\label{Definicja_para_Kahler}
A complex para-Hermite space is called complex para-K{\"a}hler space if a para-K{\"a}hler 2-form $\rho$ defined as follows
\begin{equation}
\rho (X,Y) := ds^{2} (KX, Y)
\end{equation}
is closed, $d \rho =0$.
\end{Definicja}
Definitions \ref{Definicja_para_Hermite} and \ref{Definicja_para_Kahler} can be translated into the formalism of $\mathcal{C}s$. It was proved in \cite{Przanowski_Formanski_Chudecki} (see also references therein) that the following statements are equivalent
\begin{itemize}
\item $(\mathcal{M}, ds^{2}, K)$ is a complex pH-space (pK-space)
\item $(\mathcal{M}, ds^{2}, K)$ is equipped with two complementary expanding (nonexpanding) congruences of null strings 
\end{itemize}
For pHE-spaces and pKE-spaces Corollary \ref{wniosek_o_type_d_Ein} holds true. Thus, if we choose an orientation in such a manner that both $\mathcal{C}s$ are SD, then we conclude that complex pHE-spaces (pKE-spaces) are spaces of the types
\begin{equation}
\label{symbol_para_Hermite}
[\textrm{D}]^{ee} \otimes [\textrm{any}], \ \ \ ([\textrm{D}]^{nn} \otimes [\textrm{any}])
\end{equation}
where the superscript $e$ ($n$) means that $\mathcal{C}$ is expanding (nonexpanding). The double superscript $ee$ ($nn$) means that both SD $\mathcal{C}s$ are expanding (nonexpanding). Hence, the SD Weyl spinor in pHE-spaces is always of the type $[\textrm{D}]^{ee}$. Thus, by \textsl{algebraically general (special) pHE-spaces} we understand spaces for which the ASD Weyl spinor is algebraically general (special).

For further purposes we establish that "the first" SD $\mathcal{C}$ is generated by a spinor $m_{A}$ ($\mathcal{C}_{m^{A}}$, expansion $M^{\dot{A}}$) and "the second" SD $\mathcal{C}$ is generated by a spinor $n_{A}$ ($\mathcal{C}_{n^{A}}$, expansion $N^{\dot{A}}$), $m_{A} n^{A} \ne 0$. 

Similarities between the definitions of complex and real pH-spaces (compare \cite{Bor_Makhmali_Nurowski}) justify why in what follows we omit the adjective "complex".

\subsection{Congruences of null geodesics}

Let us assume an algebraic degeneracy of the ASD Weyl spinor. According to Theorem \ref{Twierdzenie_Goldberga_Sachsa} it is equivalent to an existence of an ASD $\mathcal{C}$. Let it be generated by a spinor $m_{\dot{A}}$. Hence, we denote it by $\mathcal{C}_{m^{\dot{A}}}$. Its expansion is given by a spinor $M^{A}$. We refer to this $\mathcal{C}$ as to "the third" congruence.

SD and ASD $\mathcal{C}s$ always intersect and such an intersection constitutes a congruence of null geodesics. Before we proceed further, we introduce the following abbreviations
\begin{eqnarray}
\nonumber
\mathcal{I} &-& \textrm{a congruence of null geodesics } (\mathcal{I} \textrm{ like } \mathcal{I}\textrm{ntersection})
\\ \nonumber
\mathcal{I}s &-& \textrm{congruences of null geodesics}
\\ \nonumber
\mathcal{I}(\mathcal{C}_{m^{A}},\mathcal{C}_{m^{\dot{A}}}) &-& \textrm{a congruence of null geodesics which is an intersection of } \mathcal{C}_{m^{A}} \textrm{ and } \mathcal{C}_{m^{\dot{A}}}
\end{eqnarray}

$\mathcal{C}_{m^{\dot{A}}}$ intersects with $\mathcal{C}_{m^{A}}$ and with $\mathcal{C}_{n^{A}}$. Hence, we have two different $\mathcal{I}s$ (see Figure \ref{Congruences}). We denote them by 
\begin{equation}
\mathcal{I}_{1} := \mathcal{I} (\mathcal{C}_{m^{A}}, \mathcal{C}_{m^{\dot{A}}}), \ \ 
\mathcal{I}_{3} := \mathcal{I} (\mathcal{C}_{n^{A}}, \mathcal{C}_{m^{\dot{A}}})
\end{equation}

\begin{figure}[ht]
\begin{center}
\includegraphics[scale=0.9]{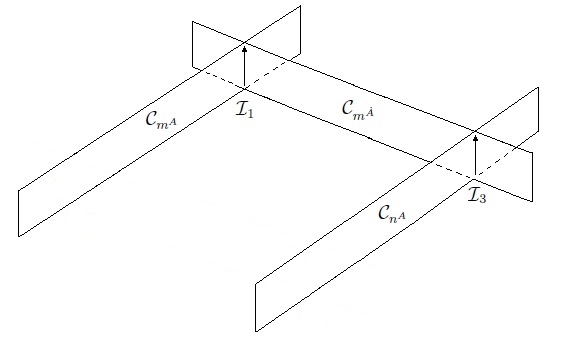}
\caption {Congruences of null strings and congruences of null geodesics in algebraically degenerate para-Hermite Einstein spaces.}
\label{Congruences}
\end{center}
\end{figure}

To classify $\mathcal{I}s$ we proceed analogously like in the general theory of relativity where congruences of null geodesics are classified according to their optical properties. Let us define \textsl{a complex expansion}\footnote{Note that the expansion of a congruence of null strings and the expansion of a congruence of null geodesics are  different concepts.} $\theta$ and \textsl{a complex twist} $\varrho$ of $\mathcal{I} (\mathcal{C}_{m^{A}}, \mathcal{C}_{m^{\dot{A}}})$ by the formulas
\begin{equation}
\label{definicje_expansji_i_twistu}
2 \theta := \nabla^{a} K_{a}, \ 2 \varrho^{2} := \nabla_{[a} K_{b]} \, \nabla^{a} K^{b}  
\end{equation}
where $K_{a} \sim m_{A} m_{\dot{A}}$ is a null vector field along the $\mathcal{I} (\mathcal{C}_{m^{A}}, \mathcal{C}_{m^{\dot{A}}})$. The following relations hold true (compare \cite{Chudecki_Ref_1})
\begin{eqnarray}
\label{I_properties_of_1and3}
\mathcal{I}_{1}: \ && \ \theta_{1} \sim m_{A} M^{A} + m_{\dot{A}} M^{\dot{A}}, \ \varrho_{1} \sim m_{A} M^{A} - m_{\dot{A}} M^{\dot{A}},
\\ \nonumber
\mathcal{I}_{3}: \ && \ \theta_{3} \sim n_{A} M^{A} + m_{\dot{A}} N^{\dot{A}}, \ \varrho_{3} \sim n_{A} M^{A} - m_{\dot{A}} N^{\dot{A}}
\end{eqnarray}
Properties of $\mathcal{I}s$ are the third of three criteria which we use for classification of pHE spaces. 

For algebraically degenerate pHE-spaces we introduce the following symbol
\begin{equation}
\label{symbol_para_Hermite_2}
\{ [\textrm{D}]^{ee} \otimes [\textrm{deg}]^{j}, [p(\mathcal{I}_{1}), p(\mathcal{I}_{3})] \}
\end{equation}
where $j= \{ n,e \}$; $p(\mathcal{I}_{i})$ stands for optical properties of $\mathcal{I}_{i}$, $p(\mathcal{I}_{i}) = \{ ++, +-, -+, -- \}$ and
\begin{eqnarray}
\label{definicja_wlasnosci_kongruencji_zerowych_geodezyjnych}
\textrm{if }\theta \ne 0, \varrho \ne 0: && p(\mathcal{I}) = [++]
\\ \nonumber
\textrm{if }\theta \ne 0, \varrho = 0: && p(\mathcal{I}) = [+-]
\\ \nonumber
\textrm{if }\theta = 0, \varrho \ne 0: && p(\mathcal{I}) = [-+]
\\ \nonumber
\textrm{if }\theta = 0, \varrho = 0: && p(\mathcal{I}) = [--]
\end{eqnarray}

The symbol (\ref{symbol_para_Hermite_2}) fails if the ASD Weyl spinor is type-[D] and there exists one more ASD $\mathcal{C}$. Let this congruence be generated by a spinor $n_{\dot{A}}$, $m_{\dot{A}} n^{\dot{A}} \ne 0$, with an expansion given by $N^{A}$. We denote this congruence by $\mathcal{C}_{n^{\dot{A}}}$ and we refer to it as to "the fourth" $\mathcal{C}$ (see Figure \ref{Congruences_4}). 
\begin{figure}[ht]
\begin{center}
\includegraphics[scale=0.9]{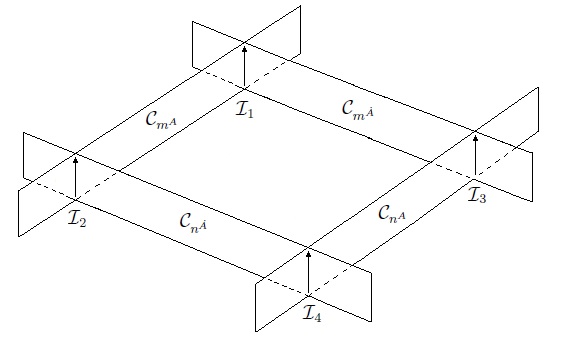}
\caption {Congruences of null strings and congruences of null geodesics in type-[D] para-Hermite Einstein spaces.}
\label{Congruences_4}
\end{center}
\end{figure}

Hence, we have two more $\mathcal{I}s$, say 
\begin{equation}
\mathcal{I}_{2} := \mathcal{I} (\mathcal{C}_{m^{A}}, \mathcal{C}_{n^{\dot{A}}}), \ \ 
\mathcal{I}_{4} := \mathcal{I} (\mathcal{C}_{n^{A}}, \mathcal{C}_{n^{\dot{A}}})
\end{equation}
with optical properties
\begin{eqnarray}
\mathcal{I}_{2}: \ && \ \theta_{2} \sim m_{A} N^{A} + n_{\dot{A}} M^{\dot{A}}, \ \varrho_{2} \sim m_{A} N^{A} - n_{\dot{A}} M^{\dot{A}},
\\ \nonumber
\mathcal{I}_{4}: \ && \ \theta_{4} \sim n_{A} N^{A} + n_{\dot{A}} N^{\dot{A}}, \ \varrho_{4} \sim n_{A} N^{A} - n_{\dot{A}} N^{\dot{A}}
\end{eqnarray}
Consequently, properties of type-[D] pHE-spaces can be gathered in the symbol
\begin{equation}
\label{symbol_para_Hermite_3}
\{ [\textrm{D}]^{ee} \otimes [\textrm{D}]^{jj}, [p(\mathcal{I}_{1}), p(\mathcal{I}_{2}), p(\mathcal{I}_{3}), p(\mathcal{I}_{4})] \}
\end{equation}
Note, that ASD Weyl spinor can be of the type $[\textrm{D}]^{ee}$ or $[\textrm{D}]^{nn}$. The "mixed" possibility $[\textrm{D}]^{en}$ in pHE-spaces cannot appear (see \cite{Przanowski_Formanski_Chudecki}, Section 3.3), unless the ASD Weyl spinor vanishes and a space is of the type $[\textrm{D}]^{ee} \otimes [\textrm{O}]$.

With all these subtleties, abbreviations and notations explained, we are ready to move on to the next step. In the next Section we recall the general structure of hyperheavenly spaces of the types $[\textrm{D}]^{ee} \otimes [\textrm{any}]$ and we specify it to the case $[\textrm{D}]^{ee} \otimes [\textrm{deg}]^{n}$.


\renewcommand{\arraystretch}{1.5}
\setlength\arraycolsep{2pt}
\setcounter{equation}{0}

\section{Algebraically degenerate para-Hermite Einstein spaces}
\label{section_HH_spaces}
\setcounter{equation}{0}

\subsection{Hyperheavenly spaces of the types $[\textrm{D}]^{ee} \otimes [\textrm{any}]$}

\subsubsection{The metric}
\label{subsubsekcja_metric}

Let us remind the definition of a \textsl{hyperheavenly space ($\mathcal{HH}$-space)} with cosmological constant $\Lambda$ \cite{Plebanski_Robinson_2}.
\begin{Definicja}
A hyperheavenly space ($\mathcal{HH}$-space) with cosmological constant $\Lambda$ is a pair $(\mathcal{M}, ds^{2})$ where $\mathcal{M}$ is a 4-dimensional complex analytic differential manifold and $ds^{2}$ is a holomorphic metric satisfying the vacuum Einstein equations with cosmological constant $\Lambda$ and such that the self-dual part of the Weyl tensor is algebraically degenerate.
\end{Definicja}
Thus, $\mathcal{HH}$-spaces are spaces of the types $[\textrm{deg}]^{e} \otimes [\textrm{any}]$ (which are called \textsl{expanding $\mathcal{HH}$-spaces}) or $[\textrm{deg}]^{n} \otimes [\textrm{any}]$ (\textsl{nonexpanding $\mathcal{HH}$-spaces}), see \cite{Plebanski_Robinson_2,Finley_Plebanski_intrinsic}. By virtue of (\ref{symbol_para_Hermite}) we find that pHE-spaces (pKE-spaces) are a special class of expanding (nonexpanding) $\mathcal{HH}$-spaces. A specification of expanding $\mathcal{HH}$-spaces to be of the types $[\textrm{D}]^{ee} \otimes [\textrm{any}]$ is straightforward so we omit all details (the Reader who is interested in this procedure should study Section 2 of the paper \cite{Chudecki_homothetic}). Eventually we find that a metric of any pHE-space can be brought to the form
\begin{equation}
\label{metryka_HH_ekspandujaca_D_any}
\frac{1}{2} ds^{2} = x^{-2} \left( dqdy-dpdx + \mathcal{A} \, dp^{2} - 2 \mathcal{Q} \, dpdq + \mathcal{B} \, dq^{2} \right)  = e^{1}e^{2} + e^{3}e^{4}
\end{equation}
where $(p,q,x,y)$ are local coordinates and $(e^{1},e^{2},e^{3},e^{4})$ is a null tetrad called \textsl{Pleba\'nski tetrad}:
\begin{equation}
\label{tetrada_Plebanskiego__HH_ekspandujaca}
e^{3} = x^{-2} dp, \ e^{1} = -x^{-2}  dq, \ e^{4} = -dx + \mathcal{A} \, dp - \mathcal{Q} \, dq,\ e^{2} = -dy + \mathcal{Q} \, dp - \mathcal{B} \, dq
\end{equation}
Functions $\mathcal{A} = \mathcal{A} (q,x,y)$, $\mathcal{Q} = \mathcal{Q} (q,x,y)$ and $\mathcal{B} = \mathcal{B} (q,x,y)$ are defined as follows
\begin{subequations}
\label{zwiazek_miedzy_Qab_i_W}
\begin{eqnarray}
\label{zwiazek_miedzy_w_i_A}
\mathcal{A} &:=& -x W_{yy} + \mu_{0} x^{3} + \frac{\Lambda}{6}
\\
\label{zwiazek_miedzy_w_i_Q}
\mathcal{Q} &:=& xW_{xy}-W_{y}
\\ 
\label{zwiazek_miedzy_w_i_B}
\mathcal{B} &:=& -xW_{xx} +2W_{x}
\end{eqnarray}
\end{subequations}
where $\mu_{0} \ne 0$ is a nonzero constant and $W=W(q,x,y)$ is a holomorphic function called \textsl{the key function} which satisfies \textsl{the expanding hyperheavenly equation}
\begin{eqnarray}
\label{HH_equation_ogolne}
W_{xx}W_{yy} - W_{xy}^{2} + \frac{2}{x} (W_{y}W_{xy} - W_{x}W_{yy}) +\frac{1}{x} W_{qy} &&
\\ \nonumber
 - \mu_{0} (x^{2} W_{xx} - 3x W_{x} + 3W) - \frac{\Lambda}{6x} W_{xx}&=&0
\end{eqnarray}
\begin{Uwaga} 
\normalfont
The $\mathcal{HH}$-equation (\ref{HH_equation_ogolne}) follows from the vacuum Einstein equations with cosmological constant $\Lambda$:
\begin{equation}
C_{ab}=0 \ \Longleftrightarrow \ C_{AB \dot{C}\dot{D}}=0, \ \mathcal{R}=- 4 \Lambda
\end{equation}
where $C_{ab}$ is the traceless Ricci tensor, $C_{AB \dot{C}\dot{D}}$ is a spinorial image of the traceless Ricci tensor and $\mathcal{R}$ is the curvature scalar. For further purposes we recall the exact form of the equations $C_{12 \dot{C}\dot{D}}=0$ and $\mathcal{R}=-4 \Lambda$:
\begin{subequations}
\label{Einstein_middle_triplet}
\begin{eqnarray}
\label{Einstein_middle_triplet_0}
\mathcal{R} = -4 \Lambda & \Longrightarrow & 2\Lambda = x^{2} (\mathcal{A}_{xx} + 2 \mathcal{Q}_{xy} + \mathcal{B}_{yy}) - 6x^{3} \partial_{x} (x^{-2} \mathcal{A} ) -6 x^{3} \partial_{y} (x^{-2} \mathcal{Q}) \ \ \ 
\\ 
\label{Einstein_middle_triplet_1}
C_{12 \dot{1}\dot{1}} =0 & \Longrightarrow &   x^{2} (\mathcal{B}_{xy}+\mathcal{Q}_{xx}) - 2x \mathcal{Q}_{x} = 0
\\
\label{Einstein_middle_triplet_2}
C_{12 \dot{1}\dot{2}} =0 & \Longrightarrow &   x^{2} (\mathcal{B}_{yy}-\mathcal{A}_{xx}) - 2x ( \mathcal{Q}_{y} - \mathcal{A}_{x})= 0
\\
\label{Einstein_middle_triplet_3}
C_{12 \dot{2}\dot{2}} =0 & \Longrightarrow &  x^{2} (\mathcal{Q}_{yy}+\mathcal{A}_{xy}) - 2x \mathcal{A}_{y} = 0
\end{eqnarray}
\end{subequations}
It can be quickly checked by a direct substitution that with (\ref{zwiazek_miedzy_Qab_i_W}) assumed Eqs. (\ref{Einstein_middle_triplet}) are identically satisfied. Hence, at the first glance there is no need to list (\ref{Einstein_middle_triplet}) here. The reason for this apparent redundancy will be explained in Section \ref{Spaces_Dee_x_degn}.
\end{Uwaga}

In Pleba\'nski tetrad the SD curvature coefficients $C^{(i)}$ take the form
\begin{equation}
C^{(1)}=C^{(2)}=C^{(4)}=C^{(5)}=0, \ C^{(3)} = -2 \mu_{0} x^{3}
\end{equation}
what justifies why the constant $\mu_{0}$ must be nonzero, otherwise a space becomes \textsl{a heavenly space}. The ASD conformal curvature coefficients $\dot{C}^{(i)}$ read
\begin{eqnarray}
\label{ASD_conformal_curvature}
\frac{1}{2x^{3}} \dot{C}^{(5)} &=&  \partial_{x}^{4} \left( W - \frac{\mu_{0}}{4} x^{2} y^{2} \right)
\\ \nonumber
\frac{1}{2x^{3}} \dot{C}^{(4)} &=&  \partial_{x}^{3}\partial_{y} \left( W - \frac{\mu_{0}}{4} x^{2} y^{2} \right)
\\ \nonumber
\frac{1}{2x^{3}} \dot{C}^{(3)} &=&  \partial_{x}^{2}\partial_{y}^{2} \left( W - \frac{\mu_{0}}{4} x^{2} y^{2} \right)
\\ \nonumber
\frac{1}{2x^{3}} \dot{C}^{(2)} &=&  \partial_{x}\partial_{y}^{3} \left( W - \frac{\mu_{0}}{4} x^{2} y^{2} \right)
\\ \nonumber
\frac{1}{2x^{3}} \dot{C}^{(1)} &=& \partial_{y}^{4} \left( W - \frac{\mu_{0}}{4} x^{2} y^{2} \right)
\end{eqnarray}

\begin{Uwaga}
\normalfont
At this point an important remark is needed. Readers who are familiar with the theory of expanding $\mathcal{HH}$-spaces  could be a little astonished after reading Section \ref{subsubsekcja_metric}. The vast majority of articles devoted to expanding $\mathcal{HH}$-spaces use two sets of variables. The first set are standard variables $(q^{\dot{A}}, p^{\dot{B}})$ written in a spinor-like convention, where $q^{\dot{A}}$ are coordinates which label SD null strings while $p^{\dot{B}}$ are coordinates on leafs of a SD congruence of null strings. The second set of variables $(w,t,\eta,\phi)$ are called \textsl{Plebański-Robinson-Finley} coordinates. The relation between  $(q^{\dot{A}}, p^{\dot{B}})$ and $(w,t,\eta,\phi)$ reads
\begin{equation}
\phi = J_{\dot{A}} p^{\dot{A}}, \ \eta = K^{\dot{A}} p_{\dot{A}}, \ w=J_{\dot{A}} q^{\dot{A}}, \ t= K^{\dot{A}} q_{\dot{A}}
\end{equation}
where $(J_{\dot{A}}, K_{\dot{B}})$, $ K^{\dot{A}} J_{\dot{A}} = \tau =\textrm{const} \ne 0$ is a basis of 1-index dotted spinors. The basis $(J_{\dot{A}}, K_{\dot{B}})$ has a deep geometrical meaning because it is adapted to the geometry of SD $\mathcal{C}$, namely, the spinor $J_{\dot{A}}$ is proportional to the expansion of SD $\mathcal{C}$.

Nevertheless, the basis $(J_{\dot{A}}, K_{\dot{B}})$ can be brought to the form $J_{\dot{1}}=1$, $J_{\dot{2}}=0$, $K_{\dot{1}}=0$, $K_{\dot{2}}=-1$, $\tau=1$ without any loss of generality. Hence, $q^{\dot{A}} = (w,t)$, $p^{\dot{A}} = (\phi, \eta)$. We take an advantage of this fact and we introduce the coordinates
\begin{equation}
q:= q^{\dot{1}} = w, \ p:= q^{\dot{2}} = t, \ x:= p^{\dot{1}} = \phi, \ y := p^{\dot{2}} = \eta
\end{equation}
The coordinates $(q,p,x,y)$ are more "compatible" with coordinates used in our previous papers \cite{Chudecki_Ref_1,Chudecki_Ref_2}. We call them \textsl{the hyperheavenly coordinates}.
\end{Uwaga}

\subsubsection{Gauge freedom}

The metric (\ref{metryka_HH_ekspandujaca_D_any}) remains invariant under the transformations 
\begin{equation}
\label{gauge}
q'=q'(q), \ p'=\lambda_{0}^{-\frac{1}{2}} p + h(q), \ x'=\lambda_{0}^{-\frac{1}{2}} x, \ y' = \frac{1}{f} (\lambda_{0}^{-1} y + \lambda_{0}^{- \frac{1}{2}} h_{q} x) + \sigma (q)
\end{equation}
where $\lambda_{0} \ne 0$ is a constant, $h=h(q)$, $\sigma=\sigma(q)$ and $f=f(q) := \dfrac{dq'}{dq}$ are arbitrary functions. Transformations (\ref{gauge}) imply the transformation formula for the key function $W$
\begin{eqnarray}
\label{gauge_for_key_function}
f^{2} \lambda_{0}^{\frac{3}{2}} W' &=& W + \frac{1}{2} \mu_{0} \lambda_{0}^{\frac{1}{2}}h_{q} \left( x^{3}y + \frac{1}{2} \lambda_{0}^{\frac{1}{2}} h_{q} x^{4} \right) - \frac{1}{3} L \, x^{3} +\frac{1}{2} \frac{f_{q}}{f} \, xy
\\ \nonumber
&&   - \frac{1}{2}  \lambda_{0}^{\frac{1}{2}} f \frac{\partial}{\partial q} \left( \frac{h_{q}}{f} \right)   x^{2}
   -\frac{\Lambda}{6} \lambda_{0}^{\frac{1}{2}} h_{q} \,  y - \left( \frac{1}{2} \lambda_{0} f \sigma_{q} + \frac{\Lambda}{12} \lambda_{0} h_{q}^{2}  \right)   x - M
\end{eqnarray}
where $L=L(q)$ and $M=M(q)$ are arbitrary gauge functions restricted by the condition
\begin{equation}
\label{zwiazek_miedzy_L_i_M}
3\mu_{0} M - f^{\frac{1}{2}} \partial_{q}^{2} (f^{-\frac{1}{2}}) + \frac{\Lambda}{3} L = 0
\end{equation}
Hence, there are four arbitrary gauge functions available at our disposal, namely $f$, $\sigma$, $h$ and $L$ (or $M$, if desired) and one constant $\lambda_{0}$. The transformation formula for $\mu_{0}$ reads
\begin{equation}
\label{transformacja_na_m}
\mu_{0}' = \lambda_{0}^{\frac{3}{2}} \mu_{0}
\end{equation}
We take an advantage of (\ref{transformacja_na_m}) and in the rest of the paper we put $\mu_{0}=1$.
\begin{Uwaga}
\normalfont
The choice $\mu_{0} = 1$ leaves us with the gauge (\ref{gauge}) restricted to the condition $\lambda_{0}=1$.
\end{Uwaga}
\begin{Uwaga}
\normalfont
Note, that $\mu_{0}$ can be, in fact, brought to an arbitrary constant value, not necessarily to $1$.
\end{Uwaga}

\subsubsection{Symmetries}
\label{Symmetries_general_approach}

Homothetic\footnote{We use the following terminology: if $\chi_{0} \ne 0$ then a corresponding vector is called \textsl{proper homothetic vector}. If $\chi_{0} =0$ then a vector is called \textsl{Killing vector}.} symmetries in expanding $\mathcal{HH}$-spaces were analyzed in \cite{Chudecki_homothetic}. It was proved that the set of ten equations $\nabla_{(a} K_{b)} = \chi_{0} g_{ab}$ can be reduced to a single equation called \textsl{the master equation}. Any homothetic vector $K$ admitted by expanding $\mathcal{HH}$-spaces of the types $[\textrm{D}]^{ee} \otimes [\textrm{any}]$ can be brought to the form
\begin{equation}
K = \widetilde{a} \, \frac{\partial}{\partial q} +  \widetilde{c} \, \frac{\partial}{\partial p}  + ( \widetilde{c}_{q} x   - \widetilde{a}_{q} y   -\widetilde{\epsilon}  ) \, \frac{\partial}{\partial y}
+\frac{2}{3} \chi_{0} \left( 2p \, \frac{\partial}{\partial p}  -  x \, \frac{\partial}{\partial x}  +y \, \frac{\partial}{\partial y}     \right)
\end{equation}
where $\widetilde{a}=\widetilde{a}(q)$, $\widetilde{\epsilon} = \widetilde{\epsilon} (q)$, $\widetilde{c}=\widetilde{c} (q)$ are functions of $q$. The master equation reads
\begin{eqnarray}
\label{master_equation}
KW &=& -2 \widetilde{a}_{q} W + \frac{1}{2} \widetilde{c}_{q} y \left( \mu_{0} x^{3} - \frac{\Lambda}{3} \right) 
\\ \nonumber
&& + \widetilde{\alpha} x^{3} + \frac{1}{2} (\widetilde{a}_{qq} \, xy -\widetilde{c}_{qq} \, x^{2} ) + \frac{1}{2} \widetilde{\epsilon}_{q} \, x + \frac{1}{6 \mu_{0}} (\widetilde{a}_{qqq} - 2 \Lambda \widetilde{\alpha})
\end{eqnarray}
where $\widetilde{\alpha} = \widetilde{\alpha} (q)$ is an arbitrary function. It is also well-known that in Einstein spaces nonzero cosmological constant and proper homothetic vectors are mutually exclusive. Hence
\begin{equation}
\label{symmetry_integrability_condition}
\Lambda \chi_{0} =0
\end{equation}
Under (\ref{gauge}) functions $\widetilde{a}$, $\widetilde{c}$, $\widetilde{\epsilon}$ and $\widetilde{\alpha}$ transform as follows
\begin{subequations}
\label{symmetry_transformation}
\begin{eqnarray}
\label{symmetry_transformation_a}
\widetilde{a}' &=& f \widetilde{a}
\\ 
\label{symmetry_transformation_c}
\widetilde{c}' &=&  \widetilde{c} + h_{q} \widetilde{a} - \frac{4}{3}\chi_{0} h
\\ 
\label{symmetry_transformation_epsilon}
\widetilde{\epsilon}' &=&  f^{-1} \widetilde{\epsilon} - \sigma \left(  \widetilde{a}_{q} - \frac{2}{3} \chi_{0} + \widetilde{a} \partial_{q} (\ln (\sigma f)) \right)
\\ 
\label{symmetry_transformation_alpha}
f^{2} \widetilde{\alpha}' &=& \widetilde{\alpha} - \frac{1}{2} \mu_{0}  \sigma f (\partial_{q}  -h_{q} \partial_{p} ) \left(  \frac{4}{3} \chi_{0} p +\widetilde{c}  + h_{q} \widetilde{a} \right)
\\ \nonumber
&& -\frac{1}{2} \mu_{0}  \widetilde{\epsilon} h_{q} - \frac{1}{3} \widetilde{a} L_{q} + \frac{2}{3} L (\chi_{0} - \widetilde{a}_{q})
\end{eqnarray}
\end{subequations}
\begin{Uwaga}
\label{remark_o_pierwszym_Killingu}
\normalfont
Because $W=W(q,x,y)$, $\mathcal{HH}$-spaces of the types $[\textrm{D}]^{ee} \otimes [\textrm{any}]$ are automatically equipped with a Killing vector
\begin{equation}
\label{pierwszy_Killing}
K_{1} = \frac{\partial}{\partial p}
\end{equation}
Thus, any pHE-space admits a 1D-algebra of infinitesimal symmetries $A_{1}$.
\end{Uwaga}

\subsubsection{The congruences $\mathcal{C}_{m^{A}}$ and $\mathcal{C}_{n^{A}}$}

It is quite easy to find the spinors which generate "the first" and "the second" $\mathcal{C}s$ and their expansions. In Pleba\'nski tetrad the following formulas hold true
\begin{subequations}
\label{ekspansja_CmA_and_CnA}
\begin{eqnarray}
\label{ekspansja_CmA}
\mathcal{C}_{m^{A}}: && m_{A}=[0,m], m \ne 0, M_{\dot{A}} = -\sqrt{2} \, \frac{m}{x} \, [1,0]
\\ 
\label{ekspansja_CnA}
\mathcal{C}_{n^{A}}: && n_{A}=[n,0], n \ne 0, N_{\dot{A}} = \sqrt{2} \, nx \, [-\mathcal{Q}, \mathcal{A}]
\end{eqnarray}
\end{subequations}
$\mathcal{C}_{m^{A}}$ is always expanding because $M_{\dot{1}} \ne 0$. It is straightforward to check that $N_{\dot{A}}=0 \ \Longrightarrow \ \mu_{0}=0$ what is a contradiction. Thus, $\mathcal{Q}$ and $\mathcal{A}$ cannot vanish simultaneously. Consequently, $\mathcal{C}_{n^{A}}$ is also always expanding.

\subsection{Hyperheavenly spaces of the types $[\textrm{D}]^{ee} \otimes [\textrm{deg}]^{e}$}

\label{ASD_congruence_section}

The next step towards our main goal is to impose an algebraic degeneracy of the ASD Weyl spinor. According to Theorem \ref{Twierdzenie_Goldberga_Sachsa} it is equivalent to an existence of $\mathcal{C}_{m^{\dot{A}}}$. The spinor $m_{\dot{A}}$ which generates $\mathcal{C}_{m^{\dot{A}}}$ is necessarily nonzero so it can be brought to the form $m_{\dot{A}} \sim [z,1]$, $z=z(q,p,x,y)$, without any loss of generality. ASD null string equations (\ref{ASD_null_strings_equations}) written explicitly read\footnote{Compare the formulas (2.24) and (2.25) of \cite{Chudecki_Ref_2} and remember, that in this paper $\phi=x$.}
\begin{subequations}
\label{kongruencja_mdotA_rownania}
\begin{eqnarray}
\label{kongruencja_mdotA_rownania_1}
&& z_{x} - zz_{y}=0
\\ 
\label{kongruencja_mdotA_rownania_2}
&& z_{q} - zz_{p} - z_{y} \mathcal{Z} + z \frac{\partial \mathcal{Z}}{\partial y} - \frac{\partial \mathcal{Z}}{\partial x} =0 , \ \mathcal{Z} := \mathcal{B} + 2z \mathcal{Q} + z^{2} \mathcal{A}
\end{eqnarray}
\end{subequations}
and
\begin{subequations}
\label{ekspansja_pierwszej_ASD_struny}
\begin{eqnarray}
\label{ekspansja_pierwszej_ASD_struny_1}
\frac{x}{\sqrt{2}} M_{1} &=& -x z_{y}  - 1
\\ 
\label{ekspansja_pierwszej_ASD_struny_2}
\frac{1}{\sqrt{2}x}  M_{2} &=&  - x z_{p}   + x z \frac{\partial}{\partial y} (\mathcal{Q} +z \mathcal{A}) - x \frac{\partial}{\partial x} (\mathcal{Q} +z \mathcal{A}) + (1 - x z_{y}) (\mathcal{Q} +z \mathcal{A})   \ \ \ \ \ \ \ 
\end{eqnarray}
\end{subequations}
Under the transformations (\ref{gauge}) the function $z$ transforms as follows
\begin{equation}
\label{transformacja_na_z_ina_w}
f z' = z - h_{q}
\end{equation}

The key function $W$ is a function of three variables only, $W=W(q,x,y)$. On the other hand, the function $z=z(q,p,x,y)$ is a function of four variables. Eqs. (\ref{kongruencja_mdotA_rownania}) contain both $W$ and $z$ and a natural question arises: if the key function depends on $(q,x,y)$ only (hence, so does the metric), what is a $p$-dependence of $z$?

\begin{Lemat} 
\label{lemat_o_niezaleznosci_z_od_p}
If $C_{\dot{A}\dot{B}\dot{C}\dot{D}} \ne 0$, then $z=z(q,x,y)$.
\end{Lemat}

\begin{proof}
Consider a case with $z_{x} \ne 0 \ \Longleftrightarrow \ z_{y} \ne 0$ first. A general solution of Eq. (\ref{kongruencja_mdotA_rownania_1}) is given in an implicit form
\begin{equation}
\label{solution_for_z}
y=-xz + \Sigma(q,p,z)
\end{equation}
where $\Sigma=\Sigma(q,p,z)$ is an arbitrary function. Note that the following relations hold true
\begin{equation}
\label{relacje_miedzy_ziwspol}
z_{x} = \frac{z}{\Sigma_{z}-x}, \ z_{y}  = \frac{1}{\Sigma_{z}-x}, \ z_{q} = -\frac{\Sigma_{q}}{\Sigma_{z}-x}, \ 
z_{p}  = -\frac{\Sigma_{p}}{\Sigma_{z}-x}
\end{equation}
Eq. (\ref{solution_for_z}) suggests that instead of the hyperheavenly coordinates $(q,p,x,y)$ one should treat $z$ as a new variable and consider an alternative coordinate system $(q,p,x,z)$. Denote
\begin{equation}
\mathcal{A} (q,x,y)= \tilde{\mathcal{A}} (q,x,y(q,p,x,z)) = \tilde{\mathcal{A}} (q,p,x,z)    
\end{equation}
Note that $\mathcal{A}$ does not depend on $p$ but its counterpart $\tilde{\mathcal{A}}$ does (as well as $\mathcal{B}$ and $\mathcal{Q}$). A general solution of Eq. (\ref{kongruencja_mdotA_rownania_2}) can be easily found in $(q,p,x,z)$ coordinate system and it reads
\begin{equation}
\label{rozwiazanie_na_Z}
\tilde{\mathcal{B}} (q,p,x,z) + 2z \tilde{\mathcal{Q}} (q,p,x,z) + z^{2} \tilde{\mathcal{A}} (q,p,x,z)  = \Omega (q,p,z) \, (x-\Sigma_{z}) + z \Sigma_{p} - \Sigma_{q} 
\end{equation}
where $\Omega= \Omega (q,p,z) $ is an arbitrary function. Differentiating Eq. (\ref{rozwiazanie_na_Z}) with respect to $x$ twice one gets (remember, that formula (\ref{rozwiazanie_na_Z}) is written in the coordinate system $(q,p,x,z)$) 
\begin{equation}
\label{rowwnanie_pomocnicze_0}
\tilde{\mathcal{B}}_{xx} + 2z \tilde{\mathcal{Q}}_{xx} + z^{2} \tilde{\mathcal{A}}_{xx}  = 0
\end{equation}
Now we transform Eq. (\ref{rowwnanie_pomocnicze_0}) to the hyperheavenly coordinate system $(q,p,x,y)$. Note that 
\begin{equation}
\tilde{\mathcal{A}}_{x} = \mathcal{A}_{x}-z \mathcal{A}_{y}, \ \tilde{\mathcal{A}}_{xx} = \mathcal{A}_{xx} - 2z \mathcal{A}_{xy} + z^{2} \mathcal{A}_{yy}, ...
\end{equation}
Hence
\begin{equation}
\label{rowwnanie_pomocnicze_1}
\mathcal{B}_{xx} + 2z ( \mathcal{Q}_{xx} - \mathcal{B}_{xy}) + z^{2} (\mathcal{B}_{yy} - 4 \mathcal{Q}_{xy} + \mathcal{A}_{xx}) + 2z^{3} (\mathcal{Q}_{yy} - \mathcal{A}_{xy}) + z^{4} \mathcal{A}_{yy} = 0
\end{equation}
Eq. (\ref{rowwnanie_pomocnicze_1}) is written in the hyperheavenly coordinates $(q,p,x,y)$; $\mathcal{A}$, $\mathcal{Q}$ and $\mathcal{B}$ do not depend on $p$ but $z$ does. Let us differentiate Eq. (\ref{rowwnanie_pomocnicze_1}) with respect to $p$
\begin{equation}
\label{rowwnanie_pomocnicze_2}
 2z_{p} ( \mathcal{Q}_{xx} - \mathcal{B}_{xy}) + 2zz_{p} (\mathcal{B}_{yy} - 4 \mathcal{Q}_{xy} + \mathcal{A}_{xx}) + 6z^{2}z_{p} (\mathcal{Q}_{yy} - \mathcal{A}_{xy}) + 4z^{3}z_{p} \mathcal{A}_{yy} = 0
\end{equation}
Assume $z_{p} \ne 0$ and differentiate (\ref{rowwnanie_pomocnicze_2}) once again with respect to $p$. We get
\begin{equation}
\label{rowwnanie_pomocnicze_3}
  2z_{p} (\mathcal{B}_{yy} - 4 \mathcal{Q}_{xy} + \mathcal{A}_{xx}) + 12zz_{p} (\mathcal{Q}_{yy} - \mathcal{A}_{xy}) + 12z^{2}z_{p} \mathcal{A}_{yy} = 0
\end{equation}
Repeating this procedure twice we find that
\begin{equation}
\label{rowwnanie_pomocnicze_4}
   12z_{p} (\mathcal{Q}_{yy} - \mathcal{A}_{xy}) + 24zz_{p} \mathcal{A}_{yy} = 0
\end{equation}
and
\begin{equation}
\label{rowwnanie_pomocnicze_5}
    24z_{p} \mathcal{A}_{yy} = 0
\end{equation}
Hence, from (\ref{rowwnanie_pomocnicze_1})-(\ref{rowwnanie_pomocnicze_5}) one gets
\begin{equation}
\label{rownania_wynikajace_z_rozniczkowan}
\mathcal{A}_{yy}=\mathcal{Q}_{yy} - \mathcal{A}_{xy}=\mathcal{B}_{yy} - 4 \mathcal{Q}_{xy} + \mathcal{A}_{xx}=\mathcal{Q}_{xx} - \mathcal{B}_{xy}=\mathcal{B}_{xx}=0
\end{equation}
Eqs. (\ref{rownania_wynikajace_z_rozniczkowan}), (\ref{Einstein_middle_triplet}) and (\ref{zwiazek_miedzy_Qab_i_W}) lead to the following form of the key function
\begin{equation}
\label{prosta_w_funkcja}
W = \frac{1}{4} \mu_{0} x^{2}y^{2} + \mathcal{P}
\end{equation}
where $\mathcal{P}$ is a third order polynomial in $x$ and $y$ with coefficients depending on $q$ only. However, such a key function implies $C_{\dot{A}\dot{B}\dot{C}\dot{D}}=0$ (compare (\ref{ASD_conformal_curvature})). Thus we arrive at a contradiction. Hence, $z_{p}=0$.

A simpler case $z_{x} = 0 \ \Longleftrightarrow \ z_{y} = 0$ can be treated similarly. We arrive at the formula (\ref{prosta_w_funkcja}) with $\mu_{0}=0$ and $\mathcal{P}$ being a second order polynomial in $x$ and $y$ with coefficients depending on $q$ only. Such a key function implies $C_{\dot{A}\dot{B}\dot{C}\dot{D}}=0$ and $C_{ABCD}=0$ what is "a double" contradiction. Hence, in this case we also find that $z_{p}=0$.
\end{proof}
\begin{Wniosek} The condition $z_{p}=0$ implies $\Sigma_{p}=0$ and $\Omega_{p}=0$.
\end{Wniosek}
\begin{proof}
$\Sigma_{p}=0$ follows from (\ref{relacje_miedzy_ziwspol}). To see that $\Omega_{p}=0$ we write (\ref{rozwiazanie_na_Z}) in the coordinate system $(q,p,x,y)$
\begin{equation}
\label{rozwiazanie_na_Z_hyp_coordinates}
\mathcal{B} (q,x,y) + 2z \mathcal{Q} (q,x,y) + z^{2} \mathcal{A} (q,x,y)  = \Omega (q,p,z) (x-\Sigma_{z})  - \Sigma_{q} 
\end{equation}
Differentiating (\ref{rozwiazanie_na_Z_hyp_coordinates}) with respect to $p$ and remembering that $z_{p}=0$ one gets $\Omega_{p}=0$.
\end{proof}
Let us summarize the results of this Section. We proved that if $z_{x} \ne 0 \ \Longleftrightarrow \ z_{y} \ne 0$ general solutions of Eqs. (\ref{kongruencja_mdotA_rownania}) read
\begin{subequations}
\begin{eqnarray}
\label{solution_for_z_precise}
&& y=-xz + \Sigma(q,z)
\\ 
\label{solution_for_Z_precise}
&& \mathcal{B}  + 2z \mathcal{Q}  + z^{2} \mathcal{A}  = \Omega (q,z) (x-\Sigma_{z})  - \Sigma_{q} 
\end{eqnarray}
\end{subequations}
where $\Sigma$ and $\Omega$ are arbitrary functions of $q$ and $z$ only. Eqs. (\ref{solution_for_z_precise}) and (\ref{solution_for_Z_precise}) are a nice starting point towards types $[\textrm{D}]^{ee} \otimes [\textrm{deg}]^{e}$ spaces. However, such a goal seems a bit too ambitious. Thus, in the next Section we specify the results for slightly less generic spaces of the types $[\textrm{D}]^{ee} \otimes [\textrm{deg}]^{n}$. Before we do this, we note that by virtue of (\ref{ekspansja_CmA_and_CnA}) the formulas (\ref{I_properties_of_1and3}) take the form
\begin{eqnarray}
\label{wlasnosci_przeciec}
&& \theta_{1} \sim xz_{y} +2, \ \varrho_{1} \sim z_{y}
\\ \nonumber
&& \theta_{3} \sim -n M_{2} - zN_{\dot{2}} + N_{\dot{1}}, \ \varrho_{3} \sim -n M_{2} + zN_{\dot{2}} - N_{\dot{1}}
\end{eqnarray}

\subsection{Hyperheavenly spaces of the types $[\textrm{D}]^{ee} \otimes [\textrm{deg}]^{n}$}
\label{Spaces_Dee_x_degn}

\subsubsection{The metric, field equations, properties of $\mathcal{I}$s}

Assume that $\mathcal{C}_{m^{\dot{A}}}$ is nonexpanding. Hence, we have to solve a set of four equations: Eqs. (\ref{kongruencja_mdotA_rownania}) and Eqs. (\ref{ekspansja_pierwszej_ASD_struny}) written for $M_{A}=0$.
\begin{Lemat} 
\label{lemat_o_postaci_z}
If the ASD congruence of null strings $\mathcal{C}_{m^{\dot{A}}}$ is nonexpanding then 
\begin{equation}
\label{solution_for_z_prostrzy}
z=- \frac{y}{x}
\end{equation}
\end{Lemat}
\begin{proof}
From (\ref{kongruencja_mdotA_rownania_1}) and $M_{1}=0$ (\ref{ekspansja_pierwszej_ASD_struny_1}) one finds that necessarily $z_{x} \ne 0 \ \Longleftrightarrow \ z_{y} \ne 0$ what implies $z=\dfrac{j(q)-y}{x}$ where $j=j(q)$ is an arbitrary function. The transformation formula for $j$ reads (compare (\ref{transformacja_na_z_ina_w}))
\begin{equation}
j' = \frac{j}{ f} + \sigma
\end{equation}
Hence, $j$ can be gauged away without any loss of generality. Thus, (\ref{solution_for_z_prostrzy}) is proved.
\end{proof}
\begin{Uwaga}
\normalfont
The solution (\ref{solution_for_z_prostrzy}) leaves us with the gauge (\ref{gauge}) restricted to the condition $\sigma=0$.
\end{Uwaga}
Using standard methods for solving partial differential equations we find that solutions of Eqs. (\ref{kongruencja_mdotA_rownania_2}) and $M_{2}=0$ (\ref{ekspansja_pierwszej_ASD_struny_2}) read
\begin{subequations}
\label{Key_function_step_1}
\begin{eqnarray}
&& \mathcal{B}  + 2z \mathcal{Q}  + z^{2} \mathcal{A}   = y \, \widetilde{E}(q,z)
\\
&& \mathcal{Q} + z \mathcal{A} = y^{2} \, \widetilde{D}(q,z)
\end{eqnarray}
\end{subequations}
where $\widetilde{E}=\widetilde{E}(q,z)$ and $\widetilde{D}=\widetilde{D}(q,z)$ are arbitrary functions. The next step is to find the key function $W$. To achieve this goal one has to cleverly juggle Eqs. (\ref{Einstein_middle_triplet}) and (\ref{zwiazek_miedzy_Qab_i_W}).
\begin{Lemat}
The key function $W$ which satisfies (\ref{Key_function_step_1}), (\ref{Einstein_middle_triplet}) and (\ref{zwiazek_miedzy_Qab_i_W}) has the form
\begin{equation}
\label{Key_function_step_3}
W = \frac{\mu_{0}}{4} x^{2} y^{2} - \frac{\Lambda}{12} \frac{y^{2}}{x} - \frac{1}{2} x^{2} \, E - x^{3} \, D + g
\end{equation}
where $E=E(q,z)$, $D=D(q,z)$ and $g(q)$ are arbitrary functions and $z=-\dfrac{y}{x}$.
\end{Lemat}
\begin{proof}
From (\ref{Key_function_step_1}) we find
\begin{equation}
\label{Key_function_step_1_1}
\mathcal{Q}  = y^{2} \, \widetilde{D}(q,z)+\frac{y}{x} \mathcal{A}, \ \mathcal{B} = y \widetilde{E} + \frac{2y^{3}}{x} \widetilde{D} + \frac{y^{2}}{x^{2}} \mathcal{A}
\end{equation}
Using (\ref{Key_function_step_1_1}) to replace $\mathcal{Q}$ and $\mathcal{B}$ in (\ref{Einstein_middle_triplet}) we get a system of equations for $\mathcal{A}$ which is quite simple to solve. The result yields
\begin{equation}
\label{forma_A_1}
\mathcal{A} = y^{2}  \widetilde{D}_{z} - 2xy  \widetilde{D} + \frac{1}{2}x \, (z\widetilde{E}_{zz} + 2 \widetilde{E}_{z}) + \frac{\Lambda}{3} + a(q) \, x^{3}
\end{equation}
where $a=a(q)$ is an arbitrary function. Inserting (\ref{forma_A_1}) into (\ref{zwiazek_miedzy_Qab_i_W}) after elementary but tedious calculations we get the key function $W$ in the form
\begin{equation}
W =  \frac{\mu_{0}}{4} x^{2} y^{2} - \frac{\Lambda}{12} \frac{y^{2}}{x} + \frac{1}{2} xy \, \widetilde{E} - x^{3}  D + g(q), \ \ D_{z} := z^{2} \widetilde{D}
\end{equation}
If we define $E := z \widetilde{E}$ we arrive at (\ref{Key_function_step_3}).
\end{proof}
The transformation formulas for $D$ and $E$ read
\begin{subequations}
\begin{eqnarray}
\label{transformacja_na_D}
f^{2} D' &=& D +\frac{1}{3} L
\\ 
\label{transformacja_na_E}
f^{2}  E' &=& E + \frac{f_{q}}{f} \, z +  f \frac{\partial}{\partial q} \left( \frac{h_{q}}{f} \right)
\end{eqnarray}
\end{subequations}

From (\ref{wlasnosci_przeciec}) one finds the properties of $\mathcal{I}_{1}$ and $\mathcal{I}_{3}$. $\mathcal{I}_{1}$ is always twisting and expanding. For $\mathcal{I}_{3}$ we find that both $\theta_{3}$ and $\varrho_{3}$ are proportional to $\mathcal{Q}  + z \mathcal{A} \sim D_{z}$. Hence, we have two branches
\begin{subequations}
\begin{eqnarray}
\label{warunek_na_podtyp_pp_pp}
\textrm{Types} \ \  \{ [\textrm{D}]^{ee} \otimes [\textrm{deg}]^{n} , [++,++]   \} \ \ \textrm{for} \ \ D_{z} \ne 0
\\ 
\label{warunek_na_podtyp_pp_mm}
\textrm{Types} \ \  \{ [\textrm{D}]^{ee} \otimes [\textrm{deg}]^{n} , [++,--]   \} \ \ \textrm{for} \ \ D_{z} = 0
\end{eqnarray}
\end{subequations}

Feeding the $\mathcal{HH}$-equation (\ref{HH_equation_ogolne}) with (\ref{Key_function_step_3}), after long and tedious calculations we find the following equations for $D$ and $E$
\begin{subequations}
\label{HH_resztki}
\begin{eqnarray}
\label{HH_resztki_1}
 EE_{zzz} - E_{zzq} - 2 \Lambda D_{z} &=& 0
\\ 
\label{HH_resztki_2}
D_{zq} - E D_{zz} +  D_{z} E_{z}   &=& 0
\end{eqnarray}
\end{subequations}
The hyperheavenly metric (\ref{metryka_HH_ekspandujaca_D_any}) written in the coordinates $(q,p,x,z)$ has the form
\begin{equation}
\label{metryka_HH_ekspandujaca_D_any_in_qpxz}
\frac{1}{2} ds^{2} = x^{-2} \left\{ -dq d(xz) - dp dx + \mathcal{A} \, dp^{2} - 2 \mathcal{Q} \, dpdq + \mathcal{B} \, dq^{2} \right\}
\end{equation}
where
\begin{eqnarray}
\label{postac_ABQ}
\mathcal{A} &=& \frac{1}{2} \mu_{0} x^{3} + \frac{\Lambda}{3} + \frac{x}{2} E_{zz} + x^{2} D_{zz}
\\ \nonumber
\mathcal{Q} &=& - \left( \frac{1}{2} \mu_{0} x^{3} + \frac{\Lambda}{3} \right) z - \frac{xz}{2} E_{zz} + x^{2} (D_{z} - zD_{zz})
\\ \nonumber
\mathcal{B} &=& \left( \frac{1}{2} \mu_{0} x^{3} + \frac{\Lambda}{3} \right) z^{2} - \frac{x}{2} (2E - z^{2} E_{zz}) + x^{2} (z^{2} D_{zz} - 2z D_{z} )
\end{eqnarray}
Note, that despite the fact that the key function depends on $D$, the metric (\ref{metryka_HH_ekspandujaca_D_any_in_qpxz}), field equations (\ref{HH_resztki}) and properties of $\mathcal{I}_{1}$ and $\mathcal{I}_{3}$ depend on $D_{z}$ only. Thus, we need solutions of (\ref{HH_resztki}) for $E$ and $D_{z}$.

\subsubsection{The ASD Weyl spinor}

For further purposes it is reasonable to discuss criteria for the Petrov-Penrose types of the ASD Weyl spinor. Inserting (\ref{Key_function_step_3}) into (\ref{ASD_conformal_curvature}) one finds
\begin{eqnarray}
\label{ASD_conformal_curvature_1}
\dot{C}^{(1)} &=& -2x^{2} D_{zzzz} - x E_{zzzz}
\\ \nonumber
\dot{C}^{(2)} &=& -x E_{zzz} - \frac{y}{x} \, \dot{C}^{(1)}
\\ \nonumber
\dot{C}^{(3)} &=& -\frac{2 \Lambda}{3} + 2y E_{zzz} + \frac{y^{2}}{x^{2}} \, \dot{C}^{(1)}
\\ \nonumber
\dot{C}^{(4)} &=& \frac{2 \Lambda y}{x} - \frac{3 y^{2}}{x} E_{zzz} - \frac{y^{3}}{x^{3}} \, \dot{C}^{(1)}
\\ \nonumber
\dot{C}^{(5)} &=& - \frac{4 \Lambda y^{2}}{x^{2}} + \frac{4 y^{3}}{x^{2}} E_{zzz} + \frac{y^{4}}{x^{4}} \, \dot{C}^{(1)}
\end{eqnarray}
The ASD Weyl spinor $C_{\dot{A}\dot{B}\dot{C}\dot{D}}$ is related to the coefficients $\dot{C}^{(i)}$ by the formula
\begin{eqnarray}
\label{spinorowa_postac_ASD_Weyl}
2 C_{\dot{A}\dot{B}\dot{C}\dot{D}} &=& \dot{C}^{(1)} \, k_{\dot{A}} k_{\dot{B}} k_{\dot{C}} k_{\dot{D}} + 4 \dot{C}^{(2)} \, l_{(\dot{A}} k_{\dot{B}} k_{\dot{C}} k_{\dot{D})} + 6 \dot{C}^{(3)} \, l_{(\dot{A}} l_{\dot{B}} k_{\dot{C}} k_{\dot{D})} 
\\ \nonumber
 && + 4 \dot{C}^{(4)} \, l_{(\dot{A}} l_{\dot{B}} l_{\dot{C}} k_{\dot{D})} + \dot{C}^{(5)} \, l_{\dot{A}} l_{\dot{B}} l_{\dot{C}} l_{\dot{D}} 
 \end{eqnarray}
where $(l_{\dot{A}}, k_{\dot{B}})$ is a basis of 1-index dotted spinors normalized to 1
\begin{equation}
k^{\dot{A}} l_{\dot{A}} = 1
\end{equation}
Inserting (\ref{ASD_conformal_curvature_1}) into (\ref{spinorowa_postac_ASD_Weyl}) one gets
\begin{equation}
\label{spinorowa_postac_ASD_Weyl_special}
 2 C_{\dot{A}\dot{B}\dot{C}\dot{D}} =  \big( -  4\Lambda \, l_{(\dot{A}}l_{\dot{B}} - 4xE_{zzz} \, l_{(\dot{A}}p_{\dot{B}}
+ \dot{C}^{(1)} \, p_{(\dot{A}}p_{\dot{B}}  \big) \, p_{\dot{C}}p_{\dot{D})}, \ \ p_{\dot{A}} := k_{\dot{A}} - \frac{y}{x} \, l_{\dot{A}} \ \ \ \ \ 
\end{equation}
Obviously, $p_{\dot{A}}$ is a double Penrose spinor. It generates the ASD $\mathcal{C}$, hence, it is proportional to $m_{\dot{A}}$ (see Section \ref{ASD_congruence_section}). Choosing a spinorial basis of 1-index dotted spinors in such a manner that $l_{\dot{A}} = [i,0]$, $k_{\dot{A}} = [0, i]$, one finds $p_{\dot{A}} = i [z,1]$. 

Consider a case with $\Lambda \ne 0$ first. Eq. (\ref{spinorowa_postac_ASD_Weyl_special}) can be rewritten as follows
\begin{equation}
\label{ASD_curvature_for _Lambda_ne0}
2 C_{\dot{A}\dot{B}\dot{C}\dot{D}} = 4 \Lambda \, n^{+}_{(\dot{A}} n^{-}_{\dot{B}} p_{\dot{C}} p_{\dot{D})}
\end{equation}
where
\begin{equation}
\label{second_multiple_dotted_PS}
 n^{\pm}_{\dot{A}} = \delta^{\dot{1}}_{\dot{A}} - i \, \frac{x E_{zzz} \pm \sqrt{x^{2} E_{zzz}^{2} + \Lambda \dot{C}^{(1)}}}{2 \Lambda} \, p_{\dot{A}}
\end{equation}
As long as $x^{2} E_{zzz}^{2} + \Lambda \dot{C}^{(1)} \ne 0$ holds true, the ASD Weyl spinor is type-[II]. Otherwise, it is type-[D]. 

If $\Lambda =0$ holds true, then Eq. (\ref{spinorowa_postac_ASD_Weyl_special}) takes the form
\begin{equation}
2 C_{\dot{A}\dot{B}\dot{C}\dot{D}} =  \big( - 4xE_{zzz} \, l_{(\dot{A}}
+ \dot{C}^{(1)} \, p_{(\dot{A}}  \big) \, p_{\dot{B}} p_{\dot{C}}p_{\dot{D})}
\end{equation}
Thus, for $E_{zzz} \ne 0$ we deal with the type [III] and for $E_{zzz} = 0$ we arrive at the type [N]. Criteria are gathered in the Table \ref{Kryteria_krzywizy_ASD}.
\begin{table}[ht]
\begin{center}
\begin{tabular}{|c|c|}   \hline
Petrov-Penrose type &  Criteria     \\  \hline
$[\textrm{D}]^{ee} \otimes [\textrm{II}]^{n}$ & $\Lambda \ne 0$, ($E_{zzzz} \ne 0$ or $2 \Lambda  D_{zzzz} \ne E_{zzz}^{2}$)  \\ \hline
$[\textrm{D}]^{ee} \otimes [\textrm{D}]^{nn}$ & $\Lambda \ne 0$, $E_{zzzz} = 0$, $2 \Lambda  D_{zzzz} = E_{zzz}^{2}$  \\ \hline
$[\textrm{D}]^{ee} \otimes [\textrm{III}]^{n}$ & $\Lambda = 0$, $E_{zzz} \ne 0$   \\ \hline
$[\textrm{D}]^{ee} \otimes [\textrm{N}]^{n}$ & $\Lambda = 0$, $E_{zzz} = 0$, $D_{zzzz} \ne 0$  \\ \hline
$[\textrm{D}]^{ee} \otimes [\textrm{O}]^{n}$ & $\Lambda = 0$, $E_{zzz} = 0$ , $D_{zzzz} = 0$  \\ \hline
\end{tabular}
\caption{Criteria for the Petrov-Penrose types.}
\label{Kryteria_krzywizy_ASD}
\end{center}
\end{table}

\subsubsection{Symmetries}

Feeding (\ref{master_equation}) with (\ref{Key_function_step_3}) one gets the system of equations
\begin{subequations}
\label{rozpisane_rownanie_master}
\begin{eqnarray}
\label{rozpisane_rownanie_master_1}
&& - \widetilde{a} D_{q} + \left( \widetilde{a}_{q} z - \frac{4}{3}\chi_{0} z + \widetilde{c}_{q} \right) D_{z} + ( 2 \chi_{0} - 2 \widetilde{a}_{q} ) D + \frac{1}{2} \mu_{0} \widetilde{\epsilon} z - \widetilde{\alpha} =0 
\\
\label{rozpisane_rownanie_master_2}
&& -\frac{1}{2} \widetilde{a} E_{q} + \left( - \frac{2}{3} \chi_{0} z + \frac{1}{2} \widetilde{a}_{q} z + \frac{1}{2} \widetilde{c}_{q} \right) E_{z} + \left( \frac{2}{3} \chi_{0} - \widetilde{a}_{q} \right) E - \widetilde{\epsilon} D_{z}
\\ \nonumber
&& \ \ \ \ \ \ \ \ \  + \frac{1}{2} (\widetilde{a}_{qq} z + \widetilde{c}_{qq} ) = 0
\\ 
\label{rozpisane_rownanie_master_3}
&& \widetilde{\epsilon} E_{z} + \widetilde{\epsilon}_{q} =0
\\
\label{rozpisane_rownanie_master_4}
&& \Lambda  \widetilde{\epsilon}  =0
\\
\label{rozpisane_rownanie_master_5}
&&  \widetilde{a}_{qqq} - 2 \Lambda \widetilde{\alpha}  =0
\end{eqnarray}
\end{subequations}


\section{Types $[\textrm{D}]^{ee} \otimes [\textrm{II}]^{n}$}
\label{section_Dee_x_IIn}
\setcounter{equation}{0}

\subsection{Type $\{ [\textrm{D}]^{ee} \otimes [\textrm{II}]^{n}, [++,++] \}$}

In this Section we deal with a case for which
\begin{equation}
\label{warunki_Typ_II_pp_pp}
D_{z} \ne 0, \ \Lambda \ne 0, \ E_{zzzz} \ne 0 \ \textrm{or} \ E_{zzz}^{2} - 2\Lambda D_{zzzz} \ne 0
\end{equation}
(Compare (\ref{warunek_na_podtyp_pp_pp}) and the Table \ref{Kryteria_krzywizy_ASD}).

\subsubsection{General case}

\begin{Twierdzenie}
\label{Twierdzenie_typ_II_pp_pp}
Let $(\mathcal{M}, ds^{2})$ be an Einstein complex space of the type $\{ [\textrm{D}]^{ee} \otimes [\textrm{II}]^{n}, [++,++] \}$. Then there exists a local coordinate system $(q,p,x,w)$ such that the metric takes the form 
\begin{eqnarray}
\label{metryka_TypII_pp_pp_ostateczna}
\frac{1}{2} ds^{2} &=& x^{-2} \left\{ -dpdx -Z \, dqdx - xZ_{w} \, dq dw \right.
\\ \nonumber
&& \ \ \ \ \ \
+ \left( \frac{1}{2} \mu_{0} x^{3} + \frac{\Lambda}{3} + \frac{1}{2} x \, M_{w} + x^{2} \, \frac{Z_{ww}}{Z_{w}} \right) dp^{2} 
\\ \nonumber
&& \ \ \ \ \ \ +2 \left[ \left( \frac{1}{2} \mu_{0} x^{3} + \frac{\Lambda}{3} \right) Z + \frac{1}{2} xZ \, M_{w} - x^{2} \left( Z_{w} -  \frac{ZZ_{ww}}{Z_{w}} \right) \right] dp dq
\\ \nonumber
&& \ \ \ \ \ \ \ 
\left.   + \left[  \left( \frac{1}{2} \mu_{0} x^{3} + \frac{\Lambda}{3} \right) Z^{2} + \frac{1}{2} xZ^{2} M_{w} + x^{2} \left(  \frac{Z^{2} Z_{ww}}{Z_{w}} - 2Z Z_{w} \right)  \right] dq^{2}  \right \}
\end{eqnarray}
where $\mu_{0}=1$, $\Lambda \ne 0$ is the cosmological constant, $Z (q,w) := -\dfrac{1}{2 \Lambda} M_{q}$ and 
\begin{equation}
\label{Typ_II_pp_pp_jedyne_rownanie}
12 \Lambda M_{ww} = M_{w}^{3} + a(w) \, M_{w} + b(w)
\end{equation}
$M=M(q,w)$, $a=a(w)$ and $b=b(w)$ are holomorphic functions such that $M_{qw} \ne 0$ and
\begin{equation}
\label{Typ_II_pp_pp_jedyne_warunki}
a_{w} M_{w} + b_{w} \ne 0 \ \ \textrm{or} \ \ 12\Lambda a_{ww} - aa_{w} + 3a_{w} M_{w}^{2} + 6b_{w} M_{w} \ne 0
\end{equation}
\end{Twierdzenie}
\begin{proof}
Because $D_{z} \ne 0$ holds, Eq. (\ref{HH_resztki_2}) can be rearranged into the form
\begin{equation}
\frac{\partial}{\partial q} \left( \frac{1}{D_{z}} \right) = \frac{\partial}{\partial z} \left( \frac{E}{D_{z}} \right)
\end{equation}
Hence, there exists a function $\Omega = \Omega (q,z)$ such that
\begin{equation}
\label{definicja_Omega}
D_{z} = \frac{1}{\Omega_{z}}, \ E = \frac{\Omega_{q}}{\Omega_{z}}
\end{equation}
The second field equation (\ref{HH_resztki_1}) takes the form
\begin{equation}
\label{drugie_rownanie_2}
\Omega_{q} S_{z} - \Omega_{z} S_{q} = 2 \Lambda, \ S:= E_{zz}= \partial^{2}_{z} \left(  \frac{\Omega_{q}}{\Omega_{z}} \right)
\end{equation}
Multiplying (\ref{drugie_rownanie_2}) by $dz \wedge dq$ we get
\begin{equation}
\label{drugie_rownanie_3}
dS \wedge d \Omega = 2 \Lambda \, dz \wedge dq
\end{equation}
Now we introduce new coordinates $(q', w)$ such that $q'=q$ and $w = \Omega (q,z)$. Thus, $z=Z(q',w)$. Note that
\begin{eqnarray}
\label{coordinates_qw_typ_II}
1 = \frac{dZ}{dz} = \frac{\partial Z}{\partial q'} \frac{\partial q'}{\partial z} + \frac{\partial Z}{\partial w} \frac{\partial \Omega}{\partial z} = Z_{w} \Omega_{z} \ &\Longleftrightarrow & \ \Omega_{z} = \frac{1}{Z_{w}}
\\ \nonumber
0 = \frac{dZ}{dq} = \frac{\partial Z}{\partial q'} \frac{\partial q'}{\partial q} + \frac{\partial Z}{\partial w} \frac{\partial \Omega}{\partial q} = Z_{q'}+ Z_{w} \Omega_{q} \ &\Longleftrightarrow & \ \Omega_{q} = -\frac{Z_{q'}}{Z_{w}}
\end{eqnarray}
what implies that $D_{z}$ and $E$ can be rewritten as 
\begin{equation}
\label{wzory_na_DzinaE}
D_{z} = Z_{w} , \ E = -Z_{q'}
\end{equation}
From (\ref{drugie_rownanie_3}) we find
\begin{equation}
S_{q'} = - 2 \Lambda \, Z_{w}
\end{equation}
Hence, there exists a function $M=M(q', w)$ such that
\begin{subequations}
\begin{eqnarray}
\label{rozwiazanie_10}
&& S = M_{w}
\\
\label{definicja_z}
&& Z =- \frac{M_{q'}}{ 2 \Lambda}
\end{eqnarray}
\end{subequations}
The only equation which remains to be solved is a consistency condition between the definition of $S$ (\ref{drugie_rownanie_2}) and (\ref{rozwiazanie_10}). First we note that for an arbitrary function $F=F(q,z) = \widehat{F} (q(q'), Z(q',w)) = \widehat{F} (q',w)$ transformation rules for derivatives yield
\begin{eqnarray}
\label{transformacja_pochodnych_ffunkci}
\frac{\partial F}{\partial q} &=& \widehat{F}_{q'} + \Omega_{q} \widehat{F}_{w} = \widehat{F}_{q'} -\frac{Z_{q'}}{Z_{w}} \, \widehat{F}_{w}
\\ \nonumber
\frac{\partial F}{\partial z} &=& \Omega_{z} \widehat{F}_{w} = \frac{1}{Z_{w}} \, \widehat{F}_{w}
\end{eqnarray}
Thus,
\begin{equation}
M_{w} = \frac{1}{Z_{w}} \frac{\partial}{\partial w} \left(  \frac{1}{Z_{w}} \frac{\partial }{\partial w} (-Z_{q'}) \right)
\end{equation}
Using (\ref{definicja_z}) one gets
\begin{equation}
\partial_{q'} (M_{w}^{2}) = \partial_{q'} (4 \Lambda  \partial_{w} \ln Z_{w})
\end{equation}
Integrating this equation with respect to $q'$ and using (\ref{definicja_z}) one finds
\begin{equation}
M_{w}^{2} M_{q'w} = 4 \Lambda M_{q'ww} - \frac{1}{3} a(w) \, M_{q'w}
\end{equation}
Integrating this equation once again with respect to $q'$ and dropping primes we arrive at (\ref{Typ_II_pp_pp_jedyne_rownanie}). To write down the metric we need $D_{z}$, $E$, $D_{zz}$ and $E_{zz}$. $D_{z}$, $E$ and $E_{zz}=S$ we already have (compare (\ref{wzory_na_DzinaE}) and (\ref{rozwiazanie_10})). It is easy to see that $D_{zz} = Z_{ww} / Z_{w}$. Thus, (\ref{metryka_TypII_pp_pp_ostateczna}) is proved. After some calculations in which Eq. (\ref{Typ_II_pp_pp_jedyne_rownanie}) is intensively used, one arrives at type-[II] conditions (\ref{Typ_II_pp_pp_jedyne_warunki}). Finally, the condition $D_{z} \ne 0$ implies $M_{q'w} \ne 0$.
\end{proof}

Note, that in this case the field equations have been reduced to a single differential equation (\ref{Typ_II_pp_pp_jedyne_rownanie}) in which we recognize \textsl{an Abel equation of the first kind} for $M_{w}$. Many explicit solutions of this equation are gathered in \cite{Polyanin}. Also, an interesting approach to solutions of the Abel equation was presented in \cite{Panayotou}.

As an example we consider a case with $b(w) = 0$. If we present $a(w)$ in the form
\begin{equation}
a(w) =: 6 \Lambda \frac{H_{ww}}{H_{w}}
\end{equation}
we find a solution of (\ref{Typ_II_pp_pp_jedyne_rownanie}) 
\begin{equation}
\label{Przyklad_rozwiazania_rownania_Abela}
M(q,w) = \int \left(  \frac{6 \Lambda H_{w}}{Q - H} \right)^{\frac{1}{2}} dw, \ \ \Longrightarrow \ \ Z(q,w) = \frac{Q_{q}}{4 \Lambda} \int \left(  \frac{6 \Lambda H_{w}}{(Q - H)^{3}} \right)^{\frac{1}{2}} dw
\end{equation}
where $Q=Q(q)$ and $H=H(w)$ are arbitrary functions such that $Q_{q} \ne 0$ and $H_{w} \ne 0$.

\subsubsection{Additional symmetries}
\label{Typ_II_pp_pp_additional_symmetries}

The metric (\ref{metryka_TypII_pp_pp_ostateczna}) is equipped with the Killing vector $K_{1} = \partial_{p}$ (\ref{pierwszy_Killing}) and it does not admit any other symmetry. Indeed, we formulate the following Lemma:
\begin{Lemat}
\label{Lemat_typ_II_brak_innych_symterii}
A complex Einstein space of the type $\{ [\textrm{D}]^{ee} \otimes [\textrm{II}]^{n}, [++,++] \}$ does not admit a symmetry algebra of dimension $\geqslant 2$.
\end{Lemat}
\begin{proof}
Assume, that a second symmetry exists and is defined by a homothetic vector $K_{2}$. From (\ref{symmetry_integrability_condition}) and (\ref{rozpisane_rownanie_master_4}) one finds $\chi_{0}=0= \widetilde{\epsilon}$, so proper homothetic vector is not admitted. Assume now that $\widetilde{a}=0$. From (\ref{rozpisane_rownanie_master_2}) we get $\widetilde{c}_{q} E_{zz}=0$. The condition $E_{zz}=0$ implies $\Lambda D_{z}=0$ (compare (\ref{HH_resztki_1})) what is a contradiction. Thus, $\widetilde{c} = \textrm{const}$. However, in such a case $K_{2} \sim K_{1}$ what is a contradiction as well. Hence, $\widetilde{a} \ne 0$. 

If $\widetilde{a} \ne 0$, without any loss of generality one puts $\widetilde{a}=1$, $\widetilde{c}=0$ (compare (\ref{symmetry_transformation_a}) and (\ref{symmetry_transformation_c})). Eqs. (\ref{rozpisane_rownanie_master_1}) and (\ref{rozpisane_rownanie_master_2}) simplify considerably and they read $D_{qz} = 0 \ \Longrightarrow \ D_{z} = D_{z} (z)$ and $E_{q}=0 \ \Longrightarrow \ E=E(z)$. The field equation (\ref{HH_resztki_2}) implies $E=E_{0} D_{z}$ where $E_{0}$ is a constant. From the field equation (\ref{HH_resztki_1}) one finds $E_{0} E_{zzz} = 2 \Lambda$. Hence, $E_{0} \ne 0$ and $E_{zzzz}=0$ hold. Moreover, the factor $E_{zzz}^{2} - 2 \Lambda D_{zzzz} \equiv D_{zzzz} (E_{0} E_{zzz} - 2 \Lambda)=0$. Consequently, a space is not type-[II] anymore what is a contradiction.
\end{proof}

\subsection{Type $\{ [\textrm{D}]^{ee} \otimes [\textrm{II}]^{n}, [++,--] \}$}
\label{Sekcja_typ_II_pp_mm}

In this Section we assume that
\begin{equation}
\label{warunki_Typ_II_pp_mm}
D_{z} = 0, \ \Lambda \ne 0,  \ E_{zzz} \ne 0
\end{equation}

\subsubsection{General case}

\begin{Twierdzenie}
\label{Twierdzenie_typ_II_pp_mm}
Let $(\mathcal{M}, ds^{2})$ be an Einstein complex space of the type $\{ [\textrm{D}]^{ee} \otimes [\textrm{II}]^{n}, [++,--] \}$. Then there exists a local coordinate system $(q,p,x,w)$ such that the metric takes the form 
\begin{eqnarray}
\label{metryka_TypII_pp_mm_ostateczna}
\frac{1}{2} ds^{2} &=& x^{-2} \bigg\{ -dpdx -Z \, dqdx - xZ_{w} \, dq dw 
\\ \nonumber
&& \ \ \ \ \ \
+   \left( \frac{\mu_{0}}{2}  x^{3} + \frac{\Lambda}{3} + \frac{wx}{2}   \right) \left( dp^{2} + 2Z \, dp dq + Z^{2} \, dq^{2} \right)  \bigg\}
\end{eqnarray}
where $\mu_{0}=1$, $\Lambda \ne 0$ is the cosmological constant, $Z$ is defined as follows
\begin{equation}
\label{Typ_II_pp_mm_def_Z}
Z(q,w) := -2 \int \frac{F_{w} \, dw}{w (q+F)^{2}}
\end{equation}
and $F=F(w)$ is an arbitrary holomorphic function such that $F_{w} \ne 0$.
\end{Twierdzenie}
\begin{proof}
With $D_{z}=0$ the field equation (\ref{HH_resztki_1}) gives
\begin{equation}
\label{HH_resztki_1_2}
E \, dE_{zz} \wedge dq - dz \wedge d E_{zz} = 0
\end{equation}
We introduce a coordinate transformation $(q,z) \longrightarrow (q',w)$ such that $q'=q$ and $w = E_{zz} (q,z)$. Thus, $z=Z(q',w)$. From (\ref{HH_resztki_1_2}) one finds
\begin{equation}
\label{wzor_na_EE}
E = -Z_{q'}
\end{equation}
Analogously like in the proof of Theorem \ref{Twierdzenie_typ_II_pp_pp} (compare formulas (\ref{coordinates_qw_typ_II}) and (\ref{transformacja_pochodnych_ffunkci})) we find that a consistency condition between equation $w = E_{zz}$ and (\ref{wzor_na_EE}) yields
\begin{equation}
\label{rownanie_Liouvilla_typ_II}
w Z_{w} = - \frac{\partial}{\partial w} \left( \frac{Z_{q'w}}{Z_{w}} \right)
\end{equation}
Substituting $H := w Z_{w}$ one rewrites (\ref{rownanie_Liouvilla_typ_II}) in the form
\begin{equation}
\label{Liouville_equation}
H = - \partial_{w} \left(  \frac{H_{q'}}{H} \right)
\end{equation}
Eq. (\ref{Liouville_equation}) is a famous Liouville equation with a general solution
\begin{equation}
\label{Liouville_equation_1}
H = -\frac{2 F_{w} Q_{q'}}{(F+Q)^{2}}, \ F=F(w), \ Q=Q(q')
\end{equation}
Hence
\begin{equation}
Z  = -2 Q_{q'} \int \frac{F_{w} dw}{w (F+Q)^{2}}
\end{equation}
At this point we have everything to write down the metric. In the metric the function $Q_{q'}$ is always combined with the factor $dq'$. Because $Q_{q'} dq' = dQ$, we treat $Q$ as a new variable. Denoting $Q$ by $q$ we arrive at the metric (\ref{metryka_TypII_pp_mm_ostateczna}).
\end{proof}

\subsubsection{Additional symmetries}

\begin{Lemat}
A complex Einstein space of the type $\{ [\textrm{D}]^{ee} \otimes [\textrm{II}]^{n}, [++,--] \}$ does not admit a symmetry algebra of dimension $\geqslant 2$.
\end{Lemat}
\begin{proof}
We skip the proof due to its similarity to that of Lemma \ref{Lemat_typ_II_brak_innych_symterii}.
\end{proof}


\section{Types $[\textrm{D}]^{ee} \otimes [\textrm{D}]^{nn}$}
\label{section_Dee_x_Dnn}
\setcounter{equation}{0}

In this Section we consider a case for which
\begin{equation}
\label{warunki_Typ_D_pp_pp}
 \Lambda \ne 0, \ E_{zzzz} = 0, \ E_{zzz}^{2} - 2\Lambda D_{zzzz} = 0
\end{equation}

\begin{Lemat}
The key function $W$ which satisfies type-[D] conditions (\ref{warunki_Typ_D_pp_pp}) and field equations (\ref{HH_resztki}) has the form
\begin{eqnarray}
\label{W_dla_typu_D}
W &=& - \frac{3 a^{2}}{4 \Lambda} \frac{y^{4}}{x} + \left( \frac{ab}{\Lambda} + \frac{a}{2x} \right) y^{3} + \left( \frac{\mu_{0}}{4} x^{2} - \frac{\Lambda}{12x} - \frac{b}{2} - \frac{3ac-3 a_{q}}{2 \Lambda} x  \right) y^{2} 
\\ \nonumber
&& + \left( \frac{3 a d- b_{q}}{\Lambda} x^{2} + \frac{c}{2} x \right) y  - \frac{d}{2} x^{2} + g
\end{eqnarray}
where $\mu_{0}=1$, $\Lambda \ne 0$ is the cosmological constant, $a$, $b$, $c$, $d$ and $g$ are functions of $q$ only, which satisfy the following equations
\begin{eqnarray}
\label{rownania_pola_typ_D_ostateczne}
&& \partial_{q} (3ac - 3 a_{q} - b^{2}) = 0
\\ \nonumber
&& 6d a_{q} + 3a d_{q} - b_{qq} - c b_{q} = 0
\end{eqnarray}
\end{Lemat}
\begin{proof}
Condition $E_{zzzz}=0$ yields 
\begin{equation}
\label{postac_E_typ_D}
E=az^{3} + bz^{2} + cz +d
\end{equation} 
where $a$, $b$, $c$ and $d$ are arbitrary functions of $q$ only. Eq. $2\Lambda \, D_{zzzz} = E_{zzz}^{2}$ implies that $D$ is a fourth-order polynomial in $z$. Inserting these forms of $E$ and $D_{z}$ into Eq. (\ref{HH_resztki_1}) one finds 
\begin{equation} 
\label{postac_D_typ_D}
D = \frac{3a^{2}}{4 \Lambda} \, z^{4} + \frac{ab}{\Lambda} \, z^{3} + \frac{3ac - 3a_{q}}{2 \Lambda} \, z^{2} + \frac{3ad - b_{q}}{\Lambda} \, z 
\end{equation}
The second field equation (\ref{HH_resztki_2}) gives (\ref{rownania_pola_typ_D_ostateczne}). Putting (\ref{postac_E_typ_D}) and (\ref{postac_D_typ_D}) into (\ref{Key_function_step_3}) one gets (\ref{W_dla_typu_D}).
\end{proof}

Transformation formulas for $a$, $b$, $c$ and $d$ are crucial for a further analysis. Using (\ref{transformacja_na_E}), (\ref{transformacja_na_z_ina_w}) and remembering that the gauge (\ref{gauge}) is restricted ($\sigma =0$) one finds
\begin{eqnarray}
\label{transformacje_na_funkcje_w_typie_D}
a' &=&  f a
\\ \nonumber
 b' &=& b + \frac{3 h_{q}}{f} a'
\\ \nonumber
f c' &=& c - \frac{3h_{q}^{2}}{f} a' + 2h_{q} b' + \frac{f_{q}}{f}
\\ \nonumber
 f^{2} d' &=& d + \frac{h_{q}^{3}}{f} a' -  h_{q}^{2} b' +  f h_{q} c' +  f \frac{d}{dq} \left(  \frac{h_{q}}{f} \right)
\end{eqnarray}
Thus, we have two arbitrary gauge functions $f=f(q)$ and $h=h(q)$ at our disposal. 

The function $a(q)$ has a deep geometrical meaning. The ASD Weyl spinor is of the type [D] so "the fourth" congruence $\mathcal{C}_{n^{\dot{A}}}$ exists. Under (\ref{postac_E_typ_D}) and (\ref{postac_D_typ_D}), Eqs. (\ref{ASD_curvature_for _Lambda_ne0}) and (\ref{second_multiple_dotted_PS}) take the form
\begin{equation}
\label{forma_n_dot_A}
2 C_{\dot{A}\dot{B}\dot{C}\dot{D}} = 4 \Lambda \, n_{(\dot{A}} n_{\dot{B}} p_{\dot{C}} p_{\dot{D})}, \ n_{\dot{A}} = \left[ 1 - \frac{3 a y}{\Lambda}, \frac{3a x}{\Lambda} \right]
\end{equation}
It is straightforward but tedious to check by substituting (\ref{forma_n_dot_A}) into (\ref{ASD_null_strings_equations}) that $\mathcal{C}_{n^{\dot{A}}}$ is nonexpanding. Thus, we deal here with the type $[\textrm{D}]^{ee} \otimes [\textrm{D}]^{nn}$. Properties of $\mathcal{I}_{3}$ depend on $D_{z}$ (compare (\ref{warunek_na_podtyp_pp_pp}) and (\ref{warunek_na_podtyp_pp_mm})). Hence
\begin{eqnarray}
\label{wlasnosci_typ_D_I3}
p(\mathcal{I}_{3}) = [++] \ &\textrm{for}& \ a \ne 0
\\ \nonumber
p(\mathcal{I}_{3}) = [--] \ &\textrm{for}& \ a = 0
\end{eqnarray}
For $\mathcal{I}_{2}$ we find $\theta_{2}, \varrho_{2} \sim n_{\dot{A}} M^{\dot{A}} \sim a$. For $\mathcal{I}_{4}$ we have $\theta_{4}, \varrho_{4} \sim  n_{\dot{A}} N^{\dot{A}}$ which is nonzero as long as $\mu_{0} \ne 0$. Thus
\begin{eqnarray}
\label{wlasnosci_typ_D_I2}
p(\mathcal{I}_{2}) = [++] \ &\textrm{for}& \ a \ne 0
\\ \nonumber
p(\mathcal{I}_{2}) = [--] \ &\textrm{for}& \ a = 0
\\ \nonumber
p(\mathcal{I}_{4}) = [++] \ &\textrm{for}& \ \textrm{arbitrary } a
\end{eqnarray}
Gathering, we deal here with the following types
\begin{eqnarray}
\nonumber
\{ [\textrm{D}]^{ee} \otimes [\textrm{D}]^{nn},[++,++,++,++] \} &\textrm{ for }& a \ne 0
\\ \nonumber
\{ [\textrm{D}]^{ee} \otimes [\textrm{D}]^{nn},[++,--,--,++] \} &\textrm{ for }& a = 0
\end{eqnarray}

Now we are ready to formulate two theorems.
\begin{Twierdzenie}
Let $(\mathcal{M}, ds^{2})$ be an Einstein complex space of the type $\{ [\textrm{D}]^{ee} \otimes [\textrm{D}]^{nn}, [++,++,++,++]\}$. Then there exists a local coordinate system $(q,p,x,y)$ such that the metric takes the form 
\begin{eqnarray}
\label{metryka_TypD_ostateczna_pppppppp}
\frac{1}{2} ds^{2} &=& x^{-2} \bigg\{  dq dy - dp dx + \left(  \frac{9}{\Lambda} y^{2} - 3 y + \frac{\mu_{0}}{2} x^{3} + \frac{3 c_{0}}{\Lambda} x^{2} +  \frac{\Lambda}{3}    \right)   dp^{2}
\\ \nonumber
&& \ \ \ \ \ \ - 2 \left(  \frac{6}{\Lambda} \frac{y^{3}}{x} -3 \frac{y^{2}}{x}  +  \left( \frac{\mu_{0}}{2} x^{2} + \frac{\Lambda}{3x}   \right) y  + \frac{3 d_{0}}{\Lambda} x^{2}  \right)   dpdq
\\ \nonumber
&& \ \ \ \ \ \ \left. + \left( \frac{3 }{\Lambda} \frac{y^{4}}{x^{2}}- 2  \frac{y^{3}}{x^{2}}   + \left( \frac{\mu_{0}}{2}x + \frac{\Lambda}{3x^{2}} -\frac{3c_{0}}{\Lambda} \right) y^{2} + \left( \frac{6d_{0}}{\Lambda} x + c_{0} \right) y  -d_{0}x  \right)    dq^{2}       \right \}
\end{eqnarray}
where $\mu_{0} = 1$, $\Lambda \ne 0$ is the cosmological constant, $c_{0}$ and $d_{0}$ are constants. The metric (\ref{metryka_TypD_ostateczna_pppppppp}) is equipped with a 2D symmetry algebra $2A_{1}$ with generators
\begin{equation}
\label{symmetry_algebra_type_D_2D}
K_{1} = \partial_{p}, \ K_{2} = \partial_{q}, \ [K_{1}, K_{2}]=0
\end{equation}
\end{Twierdzenie}
\begin{proof}
With $a \ne 0$ assumed, one puts $a=1$ and $b=0$ without any loss of generality (compare (\ref{transformacje_na_funkcje_w_typie_D})). From Eqs. (\ref{rownania_pola_typ_D_ostateczne}) we find $c=c_{0} = \textrm{const}$ and $d=d_{0} = \textrm{const}$. Inserting (\ref{W_dla_typu_D}) with such a choice of $a$, $b$, $c$ and $d$ into (\ref{zwiazek_miedzy_Qab_i_W}) and (\ref{metryka_HH_ekspandujaca_D_any}) one arrives at (\ref{metryka_TypD_ostateczna_pppppppp}). 

Because $\Lambda \ne 0$ then necessarily $\chi_{0}=0$. Thus, a symmetry algebra consists of Killing vectors only. From (\ref{rozpisane_rownanie_master}) one gets the following relations
\begin{equation}
 \widetilde{a}_{q}=0, \  \widetilde{c}_{q}=0,  \ \widetilde{\alpha} = \widetilde{\epsilon}=0
\end{equation}
Hence, any Killing vector admitted by the metric (\ref{metryka_TypD_ostateczna_pppppppp}) has the form 
\begin{equation}
K = \widetilde{a}_{0} \partial_{p} + \widetilde{c}_{0} \partial_{q} 
\end{equation}
where $\widetilde{a}_{0}$ and $\widetilde{c}_{0}$ are constants. Immediately we arrive at (\ref{symmetry_algebra_type_D_2D}). 
\end{proof}

\begin{Twierdzenie}
Let $(\mathcal{M}, ds^{2})$ be an Einstein complex space of the type $\{ [\textrm{D}]^{ee} \otimes [\textrm{D}]^{nn}, [++,--,--,++]\}$. Then there exists a local coordinate system $(q,p,x,y)$ such that the metric takes the form 
\begin{eqnarray}
\label{metryka_TypD_ostateczna_ppmmmmpp}
\frac{1}{2} ds^{2} &=& x^{-2} \bigg\{  dq dy - dp dx + \left(  \frac{\mu_{0}}{2} x^{3} + \frac{\Lambda}{3} + b_{0} x \right)   dp^{2}
\\ \nonumber
&& \ \ \ \ \ \ - 2 \left(   \frac{\mu_{0}}{2} x^{3} + \frac{\Lambda}{3} + b_{0} x \right) \frac{y}{x} \,  dpdq
 + \left( \frac{\mu_{0}}{2} x^{3} + \frac{\Lambda}{3} \right) \frac{y^{2}}{x^{2}} \,  dq^{2}       \bigg\}
\end{eqnarray}
where $\mu_{0} = 1$, $\Lambda \ne 0$ is the cosmological constant and $b_{0}$ is an arbitrary constant. The metric (\ref{metryka_TypD_ostateczna_ppmmmmpp}) is equipped with a 4D symmetry algebra $A_{3,8} \oplus A_{1}$ (for $b_{0} \ne 0$) or $A_{4,8}$ (for $b_{0} =0$) with generators
\begin{eqnarray}
\label{symmetry_algebra_type_D_4D}
&& K_{1} = \partial_{p}, \ K_{2} = \partial_{q}, \ K_{3} = q \partial_{q} - y \partial_{y}, \ K_{4} = q \partial_{p} + x \partial_{y} + b_{0} (2qy \partial_{y} - q^{2} \partial_{q}), 
\\ \nonumber
&& [K_{1}, K_{2}]=0, \ [K_{1}, K_{3}]=0, \ [K_{1}, K_{4}]=0, 
\\ \nonumber
&& [K_{2}, K_{3}] = K_{2}, \ [K_{2}, K_{4}] = K_{1} - 2 b_{0} K_{3}, \ [K_{3}, K_{4}]= K_{4}
\end{eqnarray}
\end{Twierdzenie}
\begin{proof}
If $a=0$ the function $b$ cannot be gauged away but $c$ and $d$ can (compare (\ref{transformacje_na_funkcje_w_typie_D})). Eqs. (\ref{rownania_pola_typ_D_ostateczne}) imply $b=b_{0} = \textrm{const}$. Inserting $a=c=d=0$ and $b=b_{0}$ into (\ref{W_dla_typu_D}) and then into (\ref{zwiazek_miedzy_Qab_i_W}) and (\ref{metryka_HH_ekspandujaca_D_any}) one proves (\ref{metryka_TypD_ostateczna_ppmmmmpp}). 

From (\ref{rozpisane_rownanie_master}) we find
\begin{equation}
 \widetilde{a}_{qqq}=0, \ \widetilde{\alpha} = \widetilde{\epsilon}=0, \ \widetilde{c}_{qq}=0, \ 2 b_{0} \widetilde{c}_{q} + \widetilde{a}_{qq}=0
\end{equation}
Consequently:
\begin{equation}
\widetilde{a} = -b_{0} m_{0} q^{2} + s_{0} q + r_{0}, \ \widetilde{c}  = m_{0} q + n_{0}
\end{equation}
where $m_{0}$, $s_{0}$, $r_{0}$ and $n_{0}$ are constants. Hence, any Killing vector admitted by the metric (\ref{metryka_TypD_ostateczna_ppmmmmpp}) has the form
\begin{equation}
K = n_{0} \partial_{p} + r_{0} \partial_{q} + s_{0} (q \partial_{q}-y \partial_{y}) + m_{0} (q \partial_{p} + x \partial_{y}) + m_{0}b_{0} (2  qy \partial_{y}- q^{2} \partial_{q}   )
\end{equation}
Thus, we arrive at (\ref{symmetry_algebra_type_D_4D}). If $b_{0} \ne 0$ then by substitution
\begin{equation}
e_{0} := K_{1}, \ e_{1} := K_{2}, \ e_{2} := K_{3} - \frac{1}{2b_{0}} K_{1}, \ e_{3} := \frac{1}{b_{0}} K_{4}
\end{equation}
we find that $[e_{0}, e_{i}]=0$ for $i=1,2,3$. Nonzero commutation rules read
\begin{equation}
[e_1, e_2] = e_1, \ [e_1, e_3] = -2 e_2, \ [e_2, e_3]=e_3
\end{equation}
From the Table I of \cite{Patera} we find that these commutation rules correspond to an algebra $A_{3,8} \oplus A_{1}$. If $b_{0}=0$ then by substitution
\begin{equation}
e_{1}:=-K_1, \ e_{2}:=-K_4, \ e_3:=-K_2, \ e_4 := -K_3
\end{equation}
we arrive at the commutation rules for an algebra $A_{4,8}$.
\end{proof}


\section{Types $[\textrm{D}]^{ee} \otimes [\textrm{III}]^{n}$}
\label{section_Dee_x_IIIn}
\setcounter{equation}{0}

\subsection{Type $\{ [\textrm{D}]^{ee} \otimes [\textrm{III}]^{n}, [++,++] \}$}

In this Section we deal with a case for which
\begin{equation}
\label{warunki_Typ_III_pp_pp}
D_{z} \ne 0, \ \Lambda = 0, \ E_{zzz} \ne 0 
\end{equation}

\subsubsection{General case}

\begin{Twierdzenie}
Let $(\mathcal{M}, ds^{2})$ be an Einstein complex space of the type $\{ [\textrm{D}]^{ee} \otimes [\textrm{III}]^{n}, [++,++] \}$. Then there exists a local coordinate system $(q,p,x,w)$ such that the metric takes the form 
\begin{eqnarray}
\label{metryka_TypIII_pp_pp_ostateczna}
\frac{1}{2} ds^{2} &=& x^{-2} \bigg\{ -dpdx -Z \, dqdx - xZ_{w} \, dq dw 
\\ \nonumber
&& \ \ \ \ \ \
+ \left( \frac{1}{2} \mu_{0} x^{3}  + \frac{1}{2} x \, S + x^{2} \, \frac{Z_{ww}}{Z_{w}} \right) dp^{2} 
\\ \nonumber
&& \ \ \ \ \ \ +2 \left[  \frac{1}{2} \mu_{0} x^{3}  Z + \frac{1}{2} xZ \, S - x^{2} \left( Z_{w} -  \frac{ZZ_{ww}}{Z_{w}} \right) \right] dp dq
\\ \nonumber
&& \ \ \ \ \ \ \ 
   + \left[  \frac{1}{2} \mu_{0} x^{3}  Z^{2} + \frac{1}{2} xZ^{2} S + x^{2} \left(  \frac{Z^{2} Z_{ww}}{Z_{w}} - 2Z Z_{w} \right)  \right] dq^{2}  \bigg\}
\end{eqnarray}
where $\mu_{0}=1$, $S=S(w)$ and $F=F(w)$ are arbitrary holomorphic functions such that $S_{w} \ne 0$ and $F_{w} \ne 0$ and
\begin{equation}
\label{Typ_III_pp_pp_def_Z}
Z(q,w) := -2 \int \frac{F_{w} \, dw}{S (q+F)^{2}}
\end{equation}
\end{Twierdzenie}
\begin{proof}
The proof is similar to that of Theorems \ref{Twierdzenie_typ_II_pp_pp} and \ref{Twierdzenie_typ_II_pp_mm} but with subtle differences. Because $D_{z} \ne 0$ holds, from Eq. (\ref{HH_resztki_2}) we find 
\begin{equation}
\label{definicja_Omega_2}
D_{z} = \frac{1}{\Omega_{z}}, \ E = \frac{\Omega_{q}}{\Omega_{z}}
\end{equation}
where $\Omega = \Omega (q,z)$. From Eq. (\ref{HH_resztki_1}) one gets
\begin{equation}
\label{drugie_rownanie_2_typIII}
dS \wedge d \Omega = 0, \  S:= E_{zz}= \partial^{2}_{z} \left(  \frac{\Omega_{q}}{\Omega_{z}} \right)
\end{equation}
As before, we introduce new coordinates $(q', w)$ such that $q'=q$ and $w = \Omega (q,z)$ what implies that $z=Z(q',w)$. From (\ref{drugie_rownanie_2_typIII}) we find that $S=S(w)$. Moreover (compare (\ref{coordinates_qw_typ_II})), $D_{z} = Z_{w}$ and $E = -Z_{q'}$.

We are left with a consistency condition $E_{zz} = S$. It yields the following equation (compare (\ref{transformacja_pochodnych_ffunkci}))
\begin{equation}
\label{typ_III_rownanie}
S = \frac{1}{Z_{w}} \frac{\partial}{\partial w} \left(  -\frac{Z_{q'w}}{Z_{w}}  \right)
\end{equation}
Eq. (\ref{typ_III_rownanie}) can be simply rearrange to the Liouville-like form. Thus, a solution for $Z$ reads (compare (\ref{Liouville_equation}) and (\ref{Liouville_equation_1}))
\begin{equation}
\label{solution_for_Z_typ_III}
Z = -2 Q_{q'} \int \frac{F_{w} \, dw}{S(F+Q)^{2}}
\end{equation}
With this form of $Z$ we have everything to write down the metric. Note, that $D_{zz} = Z_{ww} / Z_{w}$. Treating $Q$ as a new variable and denoting it by $q$ we arrive at (\ref{metryka_TypIII_pp_pp_ostateczna}). Finally, conditions $D_{z} \ne 0$ and $E_{zzz} \ne 0$ implies $F_{w} \ne 0$ and $S_{w} \ne 0$.
\end{proof}

\subsubsection{Solution with a 2D symmetry algebra}
\label{Typ_DxIII_pppp_dwie_symetrie}

The metric (\ref{metryka_TypIII_pp_pp_ostateczna}) is equipped with the Killing vector $K_{1} = \partial_{p}$. In this Section we find a form of the metric (\ref{metryka_TypIII_pp_pp_ostateczna}) admitting a 2D symmetry algebra. We begin by proving that the second symmetry must be defined by a proper homothetic vector.
\begin{Lemat}
\label{Lemat_typ_III_brak_symterii_Killinga}
A complex Einstein space of the type $\{ [\textrm{D}]^{ee} \otimes [\textrm{III}]^{n}, [++,++] \}$ does not admit any  Killing vector other than $K_{1} = \partial_{p}$. 
\end{Lemat}
\begin{proof}
The condition for a space to be of the type $ [\textrm{D}]^{ee} \otimes [\textrm{III}]^{n}$ is $E_{zzz} \ne 0$. We prove that if one assumes an existence of a Killing vector ($\chi_{0} =0$) other then $K_{1} = \partial_{p}$ this condition is always violated. Indeed, from (\ref{rozpisane_rownanie_master_3}) we conclude that $\widetilde{\epsilon}=0$. If $\widetilde{a} =0$ then from (\ref{rozpisane_rownanie_master_2}) one finds $\widetilde{c}_{q} = 0$. Hence, $K_{2} \sim K_{1}$. Consider now a case with $\widetilde{a} \ne 0$. In this case we put $\widetilde{a}=1$ and $\widetilde{c} =0$ without any loss of generality (compare (\ref{symmetry_transformation})). Eq. (\ref{rozpisane_rownanie_master_2}) implies $E_{q}=0$ and from the field equation (\ref{HH_resztki_1}) we find $E_{zzz} =0$ what is a contradiction. 
\end{proof}

\begin{Twierdzenie}
\label{Twierdzenie_typ_III_pp_pp_symetrie}
Let $(\mathcal{M}, ds^{2})$ be an Einstein complex space of the type $\{ [\textrm{D}]^{ee} \otimes [\textrm{III}]^{n}, [++,++] \}$ equipped with a 2D symmetry algebra $A_{2,1}$. Then there exists a local coordinate system $(q,p,x,w)$ such that the metric takes the form (\ref{metryka_TypIII_pp_pp_ostateczna})
with 
\begin{equation}
\label{Typ_III_pp_pp_def_Z_symetrie}
Z(q,w) = Z_{0} \int  \left( q-\dfrac{3}{2 \chi_{0}} \ln w \right)^{-2}  w dw
\end{equation}
where $\chi_{0} \ne 0$ and $Z_{0} \ne 0$ are constants.
\end{Twierdzenie}

\begin{proof}
Assume, that a proper homothetic vector $K_{2}$ exists. It can be always brought to the form
\begin{equation}
\label{homotetyczny_dla_III}
K_{2} = \partial_{q} + \frac{2\chi_{0}}{3}  (2p \partial_{p} - x \partial_{x} + y \partial_{y} )
\end{equation}
Indeed, for a proper homothetic vector the function $\widetilde{c}$ can be always gauged away without any loss of generality (compare (\ref{symmetry_transformation_c})). Thus, $\widetilde{c}=0$. If additionally $\widetilde{a} =0$ holds then from (\ref{rozpisane_rownanie_master_2}) one finds $zE_{z} =E$ what implies $E_{zzz}=0$. This is a contradiction. Thus, $\widetilde{a} \ne 0$. One puts $\widetilde{a} =1$ without any loss of generality what proves (\ref{homotetyczny_dla_III}). Eqs. (\ref{rozpisane_rownanie_master_1})-(\ref{rozpisane_rownanie_master_2}) take the form
\begin{eqnarray}
&& D_{qz} + \frac{2 \chi_{0}}{3} (2z D_{zz} - D_{z} ) = 0
\\ \nonumber
&& E_{q} + \frac{4 \chi_{0}}{3} (z E_{z} - E) = 0
\end{eqnarray}
Transforming these equations into the coordinate system $(q,p,x,w)$ and remembering that $D_{z} = Z_{w}$ and $E=-Z_{q}$ one finds
\begin{eqnarray}
&& \frac{\partial}{\partial{w}} \left(  \frac{Z_{q}}{Z_{w}} - \frac{4 \chi_{0}}{3} \frac{Z}{Z_{w}}  + \frac{2 \chi_{0} w}{3} \right) =0
\\ \nonumber
&& \frac{\partial}{\partial{q}} \left(  \frac{Z_{q}}{Z_{w}} - \frac{4 \chi_{0}}{3} \frac{Z}{Z_{w}} \right) =0
\end{eqnarray}
Hence
\begin{equation}
\label{rownanie_na_Z_typIII_homotetia}
Z_{q} + \frac{2 \chi_{0}}{3} w Z_{w} - \frac{4 \chi_{0}}{3} Z=0
\end{equation}
(a constant of integration has been absorbed into $w$). Inserting (\ref{Typ_III_pp_pp_def_Z}) into (\ref{rownanie_na_Z_typIII_homotetia}) we deduce 
\begin{equation}
F(w) = - \frac{3}{2 \chi_{0}} \ln w + F_{0}, \ S(w) = \frac{S_{0}}{w^{2}}
\end{equation}
where $S_{0} \ne 0$ is a constant. Incorporating $F_{0}$ into $q$ and denoting $Z_{0} := 3 / \chi_{0} S_{0}$ we arrive at (\ref{Typ_III_pp_pp_def_Z_symetrie}). The commutator reads $[K_{1}, K_{2}] = \frac{4}{3} \chi_{0} K_{1}$. Hence, an algebra is non-abelian $A_{2,1}$.
\end{proof}

\begin{Uwaga}
\normalfont
By substitution $u:= - \dfrac{3}{4 \chi_{0}} \left( q - \dfrac{3}{2 \chi_{0}} \ln w \right)^{-1}$ the formula (\ref{Typ_III_pp_pp_def_Z_symetrie}) can be rearranged to the form 
\begin{equation}
\label{zredukowana_forma_Z}
Z = - \frac{8}{9} \chi_{0}^{2} Z_{0}  e^{\frac{4 \chi_{0} q}{3}} \int e^{\frac{1}{u}}du
 \end{equation}
The integral in (\ref{zredukowana_forma_Z}) cannot be expressed in terms of elementary functions\footnote{Note that$\int e^{\frac{1}{u}}  du = u e^{\frac{1}{u}}  - Ei ( \frac{1}{u} ) + \textrm{const}$, where $Ei(x)$ is \textsl{the exponential integral}.}.
\end{Uwaga}

No other symmetries are allowed. Indeed, we already proved that no Killing vector other then $K_{1}$ is admitted (Lemma \ref{Lemat_typ_III_brak_symterii_Killinga}) and an arbitrary space admits only one proper homothetic vector. Thus, a space of the type $\{ [\textrm{D}]^{ee} \otimes [\textrm{III}]^{n}, [++,++] \}$ admits at most 2D algebra of infinitesimal symmetries.

\subsection{Type $\{ [\textrm{D}]^{ee} \otimes [\textrm{III}]^{n}, [++,--] \}$}

In this Section we assume that the following conditions hold true
\begin{equation}
\label{warunki_Typ_III_pp_mm}
D_{z} = 0, \ \Lambda = 0, \ E_{zzz} \ne 0 
\end{equation}

\subsubsection{General case}

\begin{Twierdzenie}
\label{Twierdzenie_typ_III_pp_mm}
Let $(\mathcal{M}, ds^{2})$ be an Einstein complex space of the type  $\{ [\textrm{D}]^{ee} \otimes [\textrm{III}]^{n}, [++,--] \}$. Then there exists a local coordinate system $(q,p,x,w)$ such that the metric takes the form 
\begin{eqnarray}
\label{metryka_TypIII_pp_mm_ostateczna}
\frac{1}{2} ds^{2} &=& x^{-2} \bigg\{ -dpdx -Z \, dqdx - xZ_{w} \, dq dw 
\\ \nonumber
&& \ \ \ \ \ \
+   \left( \frac{\mu_{0}}{2}  x^{3}  + \frac{wx}{2}   \right) \left( dp^{2} + 2Z \, dp dq + Z^{2} \, dq^{2} \right)  \bigg\}
\end{eqnarray}
where  $\mu_{0}=1$ and 
\begin{equation}
\label{Typ_III_pp_mm_def_Z}
Z(q,w) = -2 \int \frac{F_{w} \, dw}{w (q+F)^{2}}
\end{equation}
where $F=F(w)$ is an arbitrary holomorphic function such that $F_{w} \ne 0$.
\end{Twierdzenie}
\begin{proof}
The proof is identical like the proof of Theorem \ref{Twierdzenie_typ_II_pp_mm}.
\end{proof}

\subsubsection{Solution with a 2D symmetry algebra}

Performing an analysis identical like in the Section \ref{Typ_DxIII_pppp_dwie_symetrie} one finds that the metric (\ref{metryka_TypIII_pp_mm_ostateczna}) does not admit any Killing vector other then $K_{1} = \partial_{p}$ but it admits a proper homothetic vector of the form (\ref{homotetyczny_dla_III}).

\begin{Twierdzenie}
Let $(\mathcal{M}, ds^{2})$ be an Einstein complex space of the type $\{ [\textrm{D}]^{ee} \otimes [\textrm{III}]^{n}, [++,--] \}$ equipped with a 2D symmetry algebra $A_{2,1}$. Then there exists a local coordinate system $(q,p,x,w)$ such that the metric takes the form (\ref{metryka_TypIII_pp_mm_ostateczna})
with 
\begin{equation}
\label{Typ_III_pp_mm_def_Z_symetrie}
Z(q,w) := - \frac{3}{2 \chi_{0}} \int \left( q + \dfrac{3}{4 \chi_{0}} \ln w \right)^{-2} w^{-2} dw
\end{equation}
\end{Twierdzenie}
\begin{proof}
We skip the proof due to its similarity to that of Theorem \ref{Twierdzenie_typ_III_pp_pp_symetrie}.
\end{proof}

\begin{Uwaga}
\normalfont
By substitution $u:= - \dfrac{3}{4 \chi_{0}} \left( q + \dfrac{3}{4 \chi_{0}} \ln w \right)^{-1}$ the formula (\ref{Typ_III_pp_mm_def_Z_symetrie}) can be rearranged to the form 
\begin{equation}
\label{zredukowana_forma_Z_2}
Z = -  \frac{8 \chi_{0}}{3} e^{\frac{4 \chi_{0} q}{3}} \int e^{\frac{1}{u}}du
 \end{equation}
\end{Uwaga}


\section{Types $[\textrm{D}]^{ee} \otimes [\textrm{N}]^{n}$}
\label{section_Dee_x_Nn}
\setcounter{equation}{0}

\subsection{Type $\{ [\textrm{D}]^{ee} \otimes [\textrm{N}]^{n}, [++,++] \}$}

In this Section we deal with a case for which
\begin{equation}
\label{warunki_Typ_N_pp_pp}
D_{z} \ne 0, \ \Lambda = 0, \ E_{zzz} = 0 
\end{equation}

\subsubsection{General case}

\begin{Twierdzenie}
Let $(\mathcal{M}, ds^{2})$ be an Einstein complex space of the type $\{ [\textrm{D}]^{ee} \otimes [\textrm{N}]^{n}, [++,++] \}$. Then there exists a local coordinate system $(q,p,x,z)$ such that the metric takes the form 
\begin{eqnarray}
\label{metryka_TypN_pp_pp_ostateczna}
\frac{1}{2} ds^{2} &=& x^{-2} \bigg\{ -dpdx -dq d(xz) + \left( \frac{1}{2} \mu_{0} x^{3} + b_{0} x + (2z H - H_{t})x^{2}  \right) dp^{2} 
\\ \nonumber
&& \ \ \ \ \ \ +2 \left(  \frac{1}{2} \mu_{0} x^{3}  z + b_{0}xz + (z^{2} H - z H_{t})x^{2} \right)  dp dq + \left(  \frac{1}{2} \mu_{0} x^{3} z^{2} - x^{2} z^{2} H_{t}  \right) dq^{2}  \bigg\}
\end{eqnarray}
where $\mu_{0}=1$, $b_{0}$ is a constant, $H=H(t)$ is an arbitrary holomorphic function such that $H_{ttt} \ne 0$ and $t := \dfrac{1}{z} - b_{0} q$.
\end{Twierdzenie}
\begin{proof}
Condition $E_{zzz} =0$ implies that $E$ has the form (\ref{postac_E_typ_D}) with $a=0$. From the transformation formulas (\ref{transformacje_na_funkcje_w_typie_D}) it follows that $c$ and $d$ can be gauged away without any loss of generality. Inserting $E=bz^{2}$ into Eq. (\ref{HH_resztki_1}) one finds that $b=b_{0} = \textrm{const}$. The second field equation (\ref{HH_resztki_2}) reduces to the form
\begin{equation}
\label{zredukowane_rownanie_pola_typ_N}
D_{zq} + b_{0} (2zD_{z} - z^{2} D_{zz})=0
\end{equation}
with a general solution $D_{z} (q,z) = z^{2} H(t)$ where $t := \dfrac{1}{z} - b_{0} q$. Type-[N] condition $D_{zzzz} \ne 0$ implies $H_{ttt} \ne 0$. Inserting such forms of $E$ and $D_{z}$ into (\ref{postac_ABQ}) and (\ref{metryka_HH_ekspandujaca_D_any_in_qpxz}) one arrives at (\ref{metryka_TypN_pp_pp_ostateczna}).
\end{proof}

\subsubsection{Solution with a 2D symmetry algebra}

Inserting $E=b_{0}  z^{2}$ into Eqs. (\ref{rozpisane_rownanie_master}) and remembering that $D_{zzzz} \ne 0$ we find the following relations
\begin{subequations}
\begin{eqnarray}
\label{zredukowane_rownania_symetrii_typ_N_1}
&&\widetilde{a}_{qqq}=0, \ \widetilde{\epsilon} =0, \ \widetilde{c}_{qq} = 0, \ 2b_{0} \widetilde{c}_{q} + \widetilde{a}_{qq} = 0 , \ \chi_{0} b_{0} = 0
\\ 
\label{zredukowane_rownania_symetrii_typ_N_2}
&&-\widetilde{a} D_{qz} + \left( \widetilde{a}_{q} z - \frac{4}{3} \chi_{0} z + \widetilde{c}_{q} \right) D_{zz} + \left( \frac{2}{3} \chi_{0} - \widetilde{a}_{q} \right) D_{z} =0
\end{eqnarray}
\end{subequations}

If we assume an existence of a proper homothetic vector then $b_{0}=0$. It means that $E=0$ and $D_{z} = D_{z} (z)$ (compare (\ref{zredukowane_rownanie_pola_typ_N})). Consequently, the key function $W$ does not depend on $q$, so does the metric. Hence, an existence of a proper homothetic vector automatically implies an existence of the third Killing vector $K_{2}=\partial_{q}$ and a symmetry algebra is 3-dimensional. In this Section we deal with a 2D symmetry algebra, so the second symmetry must be generated by a Killing vector.

\begin{Twierdzenie}
Let $(\mathcal{M}, ds^{2})$ be an Einstein complex space of the type $\{ [\textrm{D}]^{ee} \otimes [\textrm{N}]^{n}, [++,++] \}$ equipped with a 2D symmetry algebra $2A_{1}$. Then there exists a local coordinate system $(q,p,x,z)$ such that the metric takes the form  
\begin{eqnarray}
\label{metryka_TypN_pp_pp_algebra_2D}
\frac{1}{2} ds^{2} &=& x^{-2} \bigg\{ -dpdx -dq d(xz) + \left( \frac{1}{2} \mu_{0} x^{3} + H x^{2}  \right) dp^{2} 
\\ \nonumber
&& \ \ \ \ \ \ +2 \left(  \frac{1}{2} \mu_{0} x^{3}  z - (H-zH_{z})x^{2} \right)  dp dq + \left(  \frac{1}{2} \mu_{0} x^{3} z^{2} + x^{2} (z^{2} H_{z} - 2z H)  \right) dq^{2}  \bigg\}
\end{eqnarray}
where $\mu_{0}=1$ and  $H=H(z)$ is an arbitrary holomorphic function such that $H_{zzz} \ne 0$. The generators of an algebra read
\begin{equation}
K_{1} = \partial_{p}, \ K_{2} = \partial_{q}, \ [K_{1}, K_{2}]=0
\end{equation}
\end{Twierdzenie}

\begin{proof}
Let us put $\chi_{0}=0$ in (\ref{zredukowane_rownania_symetrii_typ_N_2}). Eliminating $D_{zq}$ from (\ref{zredukowane_rownanie_pola_typ_N}) and (\ref{zredukowane_rownania_symetrii_typ_N_2}) we find the following equation
\begin{equation}
\label{pomocniczcze_na_symN}
r D_{zz} - r_{z} D_{z}=0, \ r(q,z) := b_{0} \widetilde{a} z^{2} - \widetilde{a}_{q}z - \widetilde{c}_{q}
\end{equation}
If we assume that $r \ne 0$ then a solution of (\ref{pomocniczcze_na_symN}) reads $D_{z} = g(q) r(q,z)$. It implies $D_{zzzz} = g(q) r_{zzz} =0$. Hence, the ASD Weyl spinor vanishes and a space is not type-[N] anymore, what is a contradiction. Thus, $r=0$ holds true. Consequently, $\widetilde{a} = \widetilde{a}_{0} = \textrm{const}$, $\widetilde{c} = \widetilde{c}_{0} = \textrm{const}$, $b_{0} \widetilde{a}_{0}=0$. Vector $K_{2}$ has the form $K_{2} = \widetilde{a}_{0} \partial_{q} + \widetilde{c}_{0} \partial_{p}$. Without any loss of generality one puts $\widetilde{c}_{0}=0$ (we already have $K_{1} = \partial_{p}$ and a linear combination of two Killing vectors with constant coefficients is still a Killing vector). One puts $\widetilde{a}_{0} =1$ what implies $b_{0}=0$. The second Killing vector reads $K_{2} = \partial_{q}$. $K_1$ and $K_2$ commute, $[K_{1}, K_{2}]=0$ so a symmetry algebra is abelian $2A_{1}$. From (\ref{zredukowane_rownanie_pola_typ_N}) we find $D_{zq}=0 \Longrightarrow D_{z} = H(z)$. Type-[N] condition yields $H_{zzz} \ne 0$. 
\end{proof}

\subsubsection{Solution with a 3D symmetry algebra}

In this Section we assume that a type-[N] space admits the third symmetry. With $b_{0}=0$, from (\ref{zredukowane_rownania_symetrii_typ_N_1}) it follows that the third homothetic vector has $\widetilde{a} = a_{0}q + const_{1}$ and $\widetilde{c} = c_{0}q + const_{2}$. Using $K_{1}=\partial_{p}$ and $K_{2}=\partial_{q}$ both $const_{1}$ and $const_{2}$ can be put zero without any loss of generality. The third homothetic vector takes the form
\begin{equation}
K_{3} = a_{0} (q \partial_{q} - y \partial_{y}) + c_{0} (q \partial_{p} + x \partial_{y}) + \frac{2}{3} \chi_{0} (2p \partial_{p} -x \partial_{x} + y \partial_{y})
\end{equation}
Eq. (\ref{zredukowane_rownania_symetrii_typ_N_2}) simplifies and yields 
\begin{equation}
\label{rownanie_na_trzecia_symetrie_typ_N}
 \left( a_{0} z - \frac{4}{3} \chi_{0} z + c_{0} \right) H_{z} + \left( \frac{2}{3} \chi_{0} - a_{0} \right) H =0, \ \ H := D_{z}
\end{equation}
If we assume $\chi_{0}=0$ (it means that $K_{3}$ is the third Killing vector), then from (\ref{rownanie_na_trzecia_symetrie_typ_N}) we find $(a_{0} z + c_{0} ) H_{zz}=0$. If $H_{zz}=0$ then the ASD Weyl spinor vanishes, while $a_{0} z + c_{0}=0$ implies $K_{3}=0$. Both these possibilities are contradictions. Thus, the third symmetry must be generated by a proper homothetic vector, $\chi_{0} \ne 0$. Using this fact we scale $K_{3}$ to the form
\begin{equation}
K_{3} = \alpha_{0} (q \partial_{q} - y \partial_{y}) + \gamma_{0} (q \partial_{p} + x \partial_{y}) + \frac{1}{2} (2p \partial_{p} -x \partial_{x} + y \partial_{y}), \ \ \alpha_{0} := \frac{3 a_{0}}{4  \chi_{0}}, \ \gamma_{0} := \frac{3 c_{0}}{4  \chi_{0}}
\end{equation}
With such abbreviations Eq. (\ref{rownanie_na_trzecia_symetrie_typ_N}) reads
\begin{equation}
\label{rownanie_na_trzecia_symetrie_typ_N_przeskalowane}
 ( 2(\alpha_{0}-1) z + 2 \gamma_{0}  ) H_{z} + ( 1-2 \alpha_{0}  ) H =0
\end{equation}

Now we are ready to formulate
\begin{Twierdzenie}
Let $(\mathcal{M}, ds^{2})$ be an Einstein complex space of the type $\{ [\textrm{D}]^{ee} \otimes [\textrm{N}]^{n}, [++,++] \}$ equipped with a 3D symmetry algebra. Then there exists a local coordinate system $(q,p,x,z)$ such that the metric takes the form  
\begin{eqnarray}
\label{metryka_TypN_pp_pp_algebra_3D}
\frac{1}{2} ds^{2} &=& x^{-2} \bigg\{ -dpdx -dq d(xz) + \left( \frac{1}{2} \mu_{0} x^{3} + H x^{2}  \right) dp^{2} 
\\ \nonumber
&& \ \ \ \ \ \ +2 \left(  \frac{1}{2} \mu_{0} x^{3}  z - (H-zH_{z})x^{2} \right)  dp dq + \left(  \frac{1}{2} \mu_{0} x^{3} z^{2} + x^{2} (z^{2} H_{z} - 2z H)  \right) dq^{2}  \bigg\}
\end{eqnarray}
where $\mu_{0}=1$. The generators of an algebra read
\begin{equation}
K_{1} = \partial_{p}, \ K_{2} = \partial_{q}, \ K_{3} = \alpha_{0} (q \partial_{q} - y \partial_{y}) + \gamma_{0} (q \partial_{p} + x \partial_{y}) + \frac{1}{2} (2p \partial_{p} -x \partial_{x} + y \partial_{y})
\end{equation} with commutation rules
\begin{equation}
\label{komutatory_dla_typu_N}
[K_{1}, K_{2}]=0, \ [K_{1}, K_{3}] =  K_{1}, \ [K_{2}, K_{3}] = \alpha_{0}K_{2} + \gamma_{0} K_{1}
\end{equation}
The function $H=H(z)$ takes the form
\begin{subequations}
\begin{eqnarray}
\nonumber
(i) && \textrm{for the algebra } A^{\alpha_{0}}_{3,5}
\\ 
\label{rozwiazanie_typ_N_1}
&& H(z) = H_{0} \left( z + \frac{\gamma_{0}}{\alpha_{0}-1} \right)^{\frac{2 \alpha_{0} -1}{2 \alpha_{0} -2}}, \ \alpha_{0} \ne \left\{  -1,0,\frac{1}{2}, 1, \frac{3}{2} \right \}, \ \gamma_{0} \textrm{ is arbitrary} \ \ \ \ \ \
\\ \nonumber
(ii) && \textrm{for the algebra } A_{3,4}
\\ 
\label{rozwiazanie_typ_N_2}
&& H(z) = H_{0} \left( z - \frac{\gamma_{0}}{2} \right)^{\frac{3}{4}}, \ \alpha_{0}=-1, \ \gamma_{0} \textrm{ is arbitrary}
\\ \nonumber
(iii) && \textrm{for the algebra } A_{3,2}
\\ 
\label{rozwiazanie_typ_N_3}
&& H(z) = H_{0} e^{\frac{z}{2 \gamma_{0}}}, \ \alpha_{0}=1, \ \gamma_{0} \ne 0
\\ \nonumber
(iv) && \textrm{for the algebra } A_{2,1} \oplus A_{1}
\\
\label{rozwiazanie_typ_N_4}
&& H(z) = H_{0} (z-\gamma_{0})^{\frac{1}{2}}, \ \alpha_{0}=0, \ \gamma_{0} \textrm{ is arbitrary}
\end{eqnarray}
\end{subequations}
For all the cases $(i)-(iv)$, $H_{0} \ne 0$.
\end{Twierdzenie}
\begin{proof}
Note that $2\alpha_{0} = 1$ or $2\alpha_{0} = 3$ implies $H_{zzz} =0$. Thus, the ASD Weyl spinor vanishes what is a contradiction. Hence, $\alpha_{0} \ne \bigg\{ \dfrac{1}{2}, \dfrac{3}{2} \bigg\}$.

Assume $\alpha_{0} \ne \{ -1,0,1 \}$. Substituting
\begin{equation}
K_1 =: e_1, \ K_2 =: \gamma_{0} \left( e_2 - \frac{1}{\alpha_{0}-1} e_1 \right), \ K_{3} =: e_3
\end{equation}
we find that the commutation rules (\ref{komutatory_dla_typu_N}) take the form
\begin{equation}
[e_1, e_2]=0, \ [e_1, e_3] = e_1, \ [e_2, e_3] = \alpha_{0} e_2
\end{equation}
what implies that a symmetry algebra is $A^{\alpha_{0}}_{3,5}$ (compare the Table I of \cite{Patera}). Eq. (\ref{rownanie_na_trzecia_symetrie_typ_N_przeskalowane}) is straightforward to be solved and the solution reads (\ref{rozwiazanie_typ_N_1}). Solution (\ref{rozwiazanie_typ_N_2}) can be obtained from (\ref{rozwiazanie_typ_N_1}) by substituting\footnote{It means that the solution (\ref{rozwiazanie_typ_N_2}) could be interpreted as a degenerate solution (\ref{rozwiazanie_typ_N_1}) and a corresponding symmetry algebra is $A^{\alpha_{0}}_{3,5}$ for $\alpha_{0} = -1$. However, Patera, Sharp and Winternitz in their distinguished work \cite{Patera} call such an algebra $A_{3,4}$.} $\alpha_{0} = -1$. 

If $\alpha_{0} = 1$ then the substitution
\begin{equation}
\label{pomocnicze_komutatory}
K_1 =: e_1, \ K_2 =: \gamma_{0}  e_2 , \ K_{3} =: e_3
\end{equation}
leads to 
\begin{equation}
[e_1, e_2]=0, \ [e_1, e_3] = e_1, \ [e_2, e_3] = e_{1} + e_2
\end{equation}
Hence, a symmetry algebra is $A_{3,2}$. Solution of Eq. (\ref{rownanie_na_trzecia_symetrie_typ_N_przeskalowane}) is immediate and (\ref{rozwiazanie_typ_N_3}) is proved. Note, that in this case $\gamma_{0} \ne 0$ holds, otherwise the ASD Weyl spinor vanishes.

Finally, if $\alpha_{0}=0$ then the substitution (\ref{pomocnicze_komutatory}) implies
\begin{equation}
[e_1, e_2]=0, \ [e_1, e_3] = e_1, \ [e_2, e_3] = e_{1} 
\end{equation}
Consequently, a symmetry algebra is $A_{2,1} \oplus A_{1}$. It proves (\ref{rozwiazanie_typ_N_4}).
\end{proof}

\subsection{Type $\{ [\textrm{D}]^{ee} \otimes [\textrm{N}]^{n}, [++,--] \}$}

Type $\{ [\textrm{D}]^{ee} \otimes [\textrm{N}]^{n}, [++,--] \}$ is characterized by the conditions $\Lambda=D_{z}=E_{zzz}=0$. However, these conditions imply vanishing of the ASD Weyl spinor. Thus, spaces of the type $\{ [\textrm{D}]^{ee} \otimes [\textrm{N}]^{n}, [++,--] \}$ do not exist.


\section{Concluding remarks and further perspectives}
\setcounter{equation}{0}

In this paper we analyzed para-Hermite Einstein spaces for which the ASD Weyl tensor is algebraically degenerate and which are equipped with a nonexpanding congruence of ASD null strings. Thus, all spaces considered in this work belong to the Walker class. The results are gathered in the Table \ref{summary}.

The metrics for all the types except one have been written down explicitly. The Einstein equations for the type $\{ [\textrm{D}]^{ee} \otimes [\textrm{II}]^{n} , [++,++]\}$ have been reduced to the equation (\ref{Typ_II_pp_pp_jedyne_rownanie}) which is a simplified Abel equation of the first kind. An Abel equation is not always integrable but a large class of solutions is known (see, e.g., \cite{Polyanin}). As an example we found the solution (\ref{Przyklad_rozwiazania_rownania_Abela}).

\begin{table}[ht]
 \footnotesize
\begin{center}
\begin{tabular}{|c|c|c|}   \hline
 Type  &   Metric   & Functions in the metric        \\ \hline \hline
 $\{ [\textrm{D}]^{ee} \otimes [\textrm{II}]^{n} , [++,++]\}$ & (\ref{metryka_TypII_pp_pp_ostateczna}) & 1 function of 2 variables restricted by Eq. (\ref{Typ_II_pp_pp_jedyne_rownanie}), $\Lambda \ne 0$   \\ \hline
 $\{ [\textrm{D}]^{ee} \otimes [\textrm{II}]^{n} , [++,--]\}$ & (\ref{metryka_TypII_pp_mm_ostateczna})  & 1 function of 1 variable, $\Lambda \ne 0$   \\ \hline
  $\{ [\textrm{D}]^{ee} \otimes [\textrm{D}]^{nn} , [++,++,++,++]\}$ & (\ref{metryka_TypD_ostateczna_pppppppp})  &  3 constants, $\Lambda \ne 0$  \\ \hline
  $\{ [\textrm{D}]^{ee} \otimes [\textrm{D}]^{nn} , [++,--,--,++]\}$ & (\ref{metryka_TypD_ostateczna_ppmmmmpp})  &   2 constants, $\Lambda \ne 0$ \\ \hline
   $\{ [\textrm{D}]^{ee} \otimes [\textrm{III}]^{n} , [++,++]\}$ & (\ref{metryka_TypIII_pp_pp_ostateczna})  &  2 functions of 1 variable  \\ \hline
   $\{ [\textrm{D}]^{ee} \otimes [\textrm{III}]^{n} , [++,--]\}$ & (\ref{metryka_TypIII_pp_mm_ostateczna})  &  1 function of 1 variable  \\ \hline
   $\{ [\textrm{D}]^{ee} \otimes [\textrm{N}]^{n} , [++,++]\}$ & (\ref{metryka_TypN_pp_pp_ostateczna})  &   1 function of 1 variable, 1 constant  \\ \hline
   $\{ [\textrm{D}]^{ee} \otimes [\textrm{N}]^{n} , [++,--]\}$ & \multicolumn{2}{|c|}{ does not exist }    \\ \hline
\end{tabular}
\caption{Summary of the results.}
\label{summary}
\end{center}
\end{table}

Although a physical interpretation of pHE-spaces is not as clear as an interpretation of pKE-spaces, we believe that the solutions presented in this paper sooner or later will found their place in a landscape of theoretical physics. Definitely, the solutions will be interesting for mathematicians who work on 4-dimensional non-Lorentzian geometries.

What more can be done in the subject? The natural question arises: can our approach be used in more general pHE-spaces for which an ASD congruence of null strings is expanding? In other words: can vacuum Einstein equations for spaces of the types $ [\textrm{D}]^{ee} \otimes [\textrm{deg}]^{e}$ be completely solved? The answer to this question remains unknown although we suspect it is negative. The crucial question is if Eqs. (\ref{solution_for_z_precise})-(\ref{solution_for_Z_precise}) allow to find an explicit form of the key function. If "yes" then the procedure presented in the current paper could be, at least in principles, repeated. If "no", an approach to the types $[\textrm{D}]^{ee} \otimes [\textrm{deg}]^{e}$ demands an alternative method or some additional "ad hoc" assumption. Such an assumption was made in \cite{Chudecki_Examples}. It allowed to find examples of pHE-spaces of the types $ [\textrm{D}]^{ee} \otimes [\textrm{deg}]^{e}$, but we were not able to explain what was a geometrical interpretation of this assumption. Thus, a classification of the examples presented in \cite{Chudecki_Examples} is not complete.

However, there is a promising direction for further investigations with a clear geometrical interpretation. Its basic assumption is vanishing of the twist of $\mathcal{I}_{1}$, $\varrho_{1}=0$ (compare (\ref{wlasnosci_przeciec})). Equivalently, 3D-distribution spanned by both $\mathcal{C}_{m^{A}}$ and $\mathcal{C}_{m^{\dot{A}}}$ is integrable. The condition $\varrho_{1}=0$ implies $z=z(q)$ (compare (\ref{wlasnosci_przeciec})-(\ref{kongruencja_mdotA_rownania_1})). Consequently, one puts $z=0$ without any loss of generality (compare the transformation formula (\ref{transformacja_na_z_ina_w})). If $z=0$ a form of the key function can be found and the $\mathcal{HH}$-equation can be reduced to a relatively simple form. Also, $\mathcal{C}_{m^{\dot{A}}}$ is necessarily expanding (compare (\ref{ekspansja_pierwszej_ASD_struny_1})). Consequently, such spaces are not Walker spaces anymore. This case is now intensively studied, but results will be published elsewhere.

\end{document}